\documentclass[letterpaper,11pt,reqno]{article}
\pdfoutput=1
\usepackage{amsmath,amssymb,amsthm,amsfonts}
\usepackage[ttscale=0.8]{libertine}
\usepackage[libertine]{newtxmath}
\usepackage[T1]{fontenc}

\usepackage[margin=1in]{geometry}
\title{On complex roots of the independence polynomial}
\usepackage{etoolbox}

\usepackage{xspace}
\usepackage{graphicx}
\usepackage{algorithm,algorithmicx,algpseudocode}
\usepackage{tikz}

\usepackage[font=small,labelfont=bf]{caption}

\usepackage{graphicx}
\usepackage[normalem]{ulem}

\usetikzlibrary{arrows,quotes,angles}

\usepackage{hyperref}
\hypersetup{
  colorlinks=true,
  linkcolor=blue,
  linktoc=all,
  citecolor=green!30!black,
  bookmarksnumbered=true,
}

\usepackage{cleveref}

\newtheorem{theorem}{Theorem}[section]
\newtheorem{observation}[theorem]{Observation}

\newtheorem{lemma}[theorem]{Lemma}

\theoremstyle{definition}
\newtheorem{definition}[theorem]{Definition}
\newtheorem{remark}[theorem]{Remark}

\newtheorem*{question*}{Question}

\numberwithin{equation}{section}

\crefname{theorem}{Theorem}{Theorems}
\crefname{observation}{Observation}{Observations}
\crefname{claim}{Claim}{Claims}
\crefname{condition}{Condition}{Conditions}
\crefname{example}{Example}{Examples}
\crefname{fact}{Fact}{Facts}
\crefname{lemma}{Lemma}{Lemmas}
\crefname{corollary}{Corollary}{Corollaries}
\crefname{definition}{Definition}{Definitions}
\crefname{remark}{Remark}{Remarks}
\crefname{proposition}{Proposition}{Propositions}
\crefname{section}{Section}{Sections}

\newcommand{\abs}[1]{\lvert #1 \rvert}

\newcommand{\defeq}{\,:=\,}                                     \newcommand{\ceil}[1]{\left\lceil #1 \right\rceil}

\newcommand{\bC}{\mathbb{C}}

\usepackage[english]{babel}\usepackage[utf8]{inputenc}
\usepackage[noadjust]{cite}
\usepackage{etoolbox}

\newbool{garamond}
\setbool{garamond}{false}

\notbool{garamond}{}{

}

\usepackage{bm}
\usepackage{ifthen}
\usepackage{letltxmacro}

\usepackage{letltxmacro}

\usepackage{nicefrac}
\usepackage{xstring}
\usepackage{pgf}
\usepackage{tikz}
\usetikzlibrary{arrows,automata}

\ifbool{garamond}{
  \usepackage[urw-garamond]{mathdesign}
}{
  \usepackage{amssymb, amsfonts}
}

\usepackage[T1]{fontenc}

\renewcommand{\mathbf}[1]{\bm{#1}}

\usepackage{graphicx}

\usepackage{booktabs,siunitx}

\renewcommand{\abs}[1]{\ensuremath{\left|#1\right|}}

\newcommand{\diff}[2]{\frac{\text{d}#1}{\text{d}#2}}

\usepackage{mathtools}

\newcommand{\inb}[1]{\left\{#1\right\}}
\newcommand{\inp}[1]{\left(#1\right)}
\newcommand{\insq}[1]{\left[#1\right]}

\newcommand{\R}[0]{\ensuremath{\mathbb{R}}}

\DeclareMathOperator*{\argmax}{arg\,max}

\newcommand{\CC}{\mathbb C}
\newcommand{\cG}{\mathcal G}

\newcommand{\im}{\ensuremath{\iota}}

\begin{document}
\newcommand{\reg}[1]{\ensuremath{\text{${#1}$-good}}}
\newcommand{\conv}[1]{\mathrm{conv}\inb{#1}}

\author{Ferenc Bencs\thanks{Korteweg de Vries Institute for Mathematics,
    University of Amsterdam. Email: \texttt{ferenc.bencs@gmail.com}.} \and Péter Csikvári\thanks{Alfréd Rényi Institute of Mathematics and Eötvös Loránd
    University. Email: \texttt{peter.csikvari@gmail.com}.}\and Piyush Srivastava\thanks{Tata Institute of Fundamental Research. Email:
    \texttt{piyush.srivastava@tifr.res.in}.}  \and Jan
  Vondr\'{a}k\thanks{Stanford University. Email: \texttt{jvondrak@stanford.edu}.}
}

\date{}
\maketitle
{\let\thefootnote\relax
  \footnotetext{FB was supported by the NKFIH (National Research,
    Development and Innovation Office, Hungary) grant KKP-133921 when the
    project started. After this, he was funded by the Netherlands Organisation
    of Scientific Research (NWO): VI.Vidi.193.068.
    PC is supported by the MTA-R\'enyi Momentum Counting in Sparse
    Graphs Research Group.
    PS acknowledges support from the DAE, Government of
    India, under project no. RTI4001, from the Ramanujan Fellowship of SERB,
    from the Infosys foundation, through its support for the
    Infosys-Chandrasekharan virtual center for Random Geometry, and from Adobe
    Systems Incorporated via a gift to TIFR.  Part of this work was performed
    when the authors were visiting the Simons Institute for the Theory of
    Computing at the University of California, Berkeley.  The contents of this
    paper do not necessarily reflect the views of the funding agencies listed
    above.}
}

 \maketitle
\begin{abstract}
  The independence polynomial of a graph is the generating polynomial of all its
  independent sets.  Formally, given a graph $G$, its independence polynomial
  $Z_G(\lambda)$ is given by $\sum_{I} \lambda^{|I|}$, where the sum is over all
  independent sets $I$ of $G$.  The independence polynomial has been an
  important object of study in both combinatorics and computer science. In
  particular, the algorithmic problem of estimating $Z_G(\lambda)$ for a fixed
  positive $\lambda$ on an input graph $G$ is a natural generalization of the
  problem of counting independent sets, and its study has led to some of the
  most striking connections between computational complexity and the theory of
  phase transitions.  More surprisingly, the independence polynomial for
  negative and complex values of $\lambda$ also turns out to be related to
  problems in statistical physics and combinatorics.  In particular, the
  locations of the complex roots of the independence polynomial of bounded
  degree graphs turn out to be very closely related to the Lovász local lemma,
  and also to the questions in the computational complexity of counting.
  Consequently, the locations of such zeros have been studied in many works.  In
  this direction, it is known from the work of
  Shearer~\cite{shearer_problem_1985} and of Scott and
  Sokal~\cite{scott_repulsive_2005} -- inspired by the study of the Lovász local
  lemma -- that the independence polynomial $Z_G(\lambda)$ of a graph $G$ of
  maximum degree at most $d+1$ does not vanish provided that
  $\abs{\lambda} \leq \frac{d^d}{(d+1)^{d+1}}$.  Significant extensions of this
  result have recently been given in the case when $\lambda$ is in the
  \emph{right} half-plane (i.e., when $\Re \lambda \geq 0$) by Peters and
  Regts~\cite{peters_conjecture_2019} and Bencs and
  Csikvári~\cite{bencs18:_note}.  In this paper, our motivation is to further
  extend these results to find new zero free regions not only in the right half
  plane, but also in the \emph{left} half-plane, that is, when
  $\Re \lambda \leq 0$.

  We give new geometric criterions for establishing zero-free regions as well as
  for carrying out semi-rigorous numerical explorations.  We then provide two
  examples of the (rigorous) use of these criterions, by establishing two new
  zero-free regions in the left-half plane.  We also extend the results of Bencs
  and Csikvári~\cite{bencs18:_note} for the right half-plane using our
  framework.  By a direct application of the interpolation method of
  Barvinok~\cite{barvinok2017combinatorics}, combined with extensions due to
  Patel and Regts~\cite{patel2017deterministic}, our results also imply
  deterministic polynomial time approximation algorithms for the independence
  polynomial of bounded degree graphs in the new zero-free regions.
\end{abstract}

 \thispagestyle{empty}

\newpage
\pagestyle{plain}
\setcounter{page}{1}

\section{Introduction}
\label{sec:introduction}
The independence polynomial, also known as the partition function of the hard
core lattice gas in the statistical physics literature, is the graph polynomial
given by
\begin{displaymath}
  Z_G(\lambda) \defeq \sum_{I:\text{independent set in $G$}}\lambda^{|I|}.
\end{displaymath}
An independent set in a graph $G$ is subset of its vertices no two of which are
adjacent in $G$.  In statistical mechanics, the polynomial arises in the
modeling of adsorption phenomena (usually with $G$ being a lattice); while in
combinatorics, it is the natural generating function of independent sets of
graphs, and offers a natural generalization to the problem of counting
independent sets in a graph.  These connections have led to the polynomial being
studied extensively in the setting $\lambda > 0$, both in statistical physics
and in computational complexity, and, in particular, has led to some very tight
connections between the two fields~\cite{sly12,Weitz}.

The setting $\lambda < 0$, and more generally, of complex $\lambda$ is also of
interest.  In particular, the problem of understanding where the complex zeros
of $Z_G$ lie for graphs $G$ in a given class is of special interest.  In
statistical mechanics, it relates to the Yang-Lee theory of phase
transitions~\cite{leeyan52}.  In the special case when $G$ is a lattice, the
work of Dobrushin and Shlosman~\cite{DS85,DS87} also related the question to
other, more probabilistic notions of phase transitions.  In combinatorics, the
behavior of $Z_G$ at negative and complex $\lambda$ plays an important role in
the study of the Lovász local lemma; for a detailed discussion of this
connection, we refer to the work of Shearer~\cite{shearer_problem_1985} as
elucidated by Scott and Sokal~\cite{scott_repulsive_2005}.  For our purposes, we start
with the following result proved in the above two papers.  We denote by
$\cG_\Delta$ the set of finite graphs with vertex degrees at most $\Delta$, for
some fixed $\Delta \geq 3$.

\begin{theorem}[\textbf{\cite{shearer_problem_1985}, see also Corollary 5.7 and
    the discussion following it in \cite{scott_repulsive_2005}}]
  Let $d \geq 2$ be an integer. If $\lambda \in \bC$ is such that
  $\abs{\lambda} \leq \lambda^*(d) \defeq \frac{d^d}{(d + 1)^{d+1}}$ then
  $Z_G(\lambda) \neq 0$ for all graphs $G \in \cG_{d+1}$. Further, for any
  negative real $\lambda_1 < -\lambda^*(d)$, there exists a graph
  $G \in G_{d+1}$ and $\lambda'$ satisfying
  $\lambda_1 < \lambda' < -\lambda^*(d)$ such that
  $Z_G(\lambda') = 0$.\label{thm:shearer}
\end{theorem}

It is also known that the above theorem gives a full description of the
zero-free region of $Z_G$, as $G$ varies over $\cG_{d+1}$, on the negative real
line.  The emphasis in the works leading to \cref{thm:shearer} was on obtaining
zero-free regions shaped like disks (or like product of disks --
\emph{polydisks} -- in the more general setting of the multivariate independence
polynomial that we will not consider in this paper), and in the univariate
setting, \cref{thm:shearer} essentially characterizes the radius of the largest such
zero-free disk centered at the origin.  Further, two different polynomial time
approximation algorithms for $Z_G(\lambda)$ for $G \in \cG_{d+1}$ and
$\abs{\lambda} < \lambda^*(d)$ were given by Patel and
Regts~\cite{patel2017deterministic} and Harvey, Srivastava and
Vondrák~\cite{10.5555/3174304.3175407}.

\paragraph{Zero-free regions and algorithms} It is well known by now, however,
that the actual zero-free region for $Z_G$ as $G$ varies over $\cG_{d+1}$ is not
described by a disk.  It also turns out that the work towards characterizing
this region is of importance for the algorithmic problem of approximating $Z_G$
for an input graph $G \in \cG_{d+1}$.  In order to describe this connection, we
first recall the work of Peters and Regts~\cite{peters_conjecture_2019} on
proving a conjecture of Sokal.  There, they considered $Z_G$ as $G$ varies over
spherically symmetric $d$-ary trees, and proved that it is non zero as long as
$\lambda \in U_d$, where $U_d$ is the open region (see \cref{fig:d9} for an
example drawing of this curve) containing the origin bounded by the curve
\begin{equation}
  \partial U_d \defeq \inb{\left.\kappa(\alpha) \defeq \frac{-\alpha d^d}{(d + \alpha)^{d+1}}
    \right\vert \abs{\alpha} = 1}.\label{eq:12}
\end{equation}
Using the results of Peters and Regts~\cite{peters_conjecture_2019} on the
existence of zeros near the boundary of $\partial U_d$, Bezáková, Galanis,
Goldberg, and Štefankovič~\cite{bezakova_inapproximability_2018} showed that for
every complex rational $\lambda$ outside the closure of $U_d$ that does
\emph{not} lie on the positive real line, the problem of approximating (up to
any polynomial factor) $Z_G(\lambda)$ for graphs $G$ in $\cG_{d+1}$ is
\#P-hard.\footnote{In contrast, for positive real $\lambda$ outside $U_d$, the
  same problem is NP-hard~\cite{sly12,gsv}, and is unlikely to be \#P-hard for
  all such $\lambda$ unless there is a collapse in the polynomial hierarchy.
  (The fact that approximate counting with positive weights cannot be \#P-hard
  under standard complexity theoretic assumptions is a well-known direct
  consequence of Toda's theorem~\cite{toda_pp_1991} and earlier results of
  Stockmeyer~\cite{stockmeyer_complexity_1983} and
  Sipser~\cite{sipser_complexity_1983}; see, e.g., Ex.~17.5 in
  \cite{arora_barak_2009}.)}  On the other hand, due to the results of
Barvinok~\cite{barvinok2017combinatorics} and Patel and
Regts~\cite{patel2017deterministic}, the same problem admits a fully polynomial
time approximation scheme (FPTAS) for any complex rational $\lambda$ if for some
$\epsilon > 0$, the $\epsilon$-neighborhood of the line segment $[0, \lambda]$
is zero-free for the polynomials $Z_G$ for all $G \in \cG_{d+1}$.

\paragraph{Known results} In light of the above results, the problem of
characterizing the location of the zeros of $Z_G$ for general graph in
$\cG_{d+1}$ becomes of interest. Recall that $U_d$ is the zero-free region for
spherically symmetric $d$-ary trees.  Perhaps the first natural question to ask
is whether the region $U_d$ is zero-free for $Z_G$ even as $G$ varies over
\emph{all} graphs in $\cG_{d+1}$.  The answer to this is no:
Buys~\cite{buys_location_2019} showed that one can obtain a counterexample for
$3\le d+1\le 9$ by considering spherically symmetric trees in which the arity of
each vertex depends upon the distance from the root of the tree.  Thus, the
location of zeros of $Z_G$ for $G$ in $\cG_{d+1}$ inside the region $U_d$ needs
to be studied more closely.

In preparation for stating the contributions of this paper, we now turn to
describing what is known about the zero-free region of $Z_G$ for graphs in
$\cG_{d+1}$.  Let $\lambda_c(d) \defeq \frac{d^d}{(d - 1)^{d + 1}}$ be the
unique point of intersection of the curve $\partial U_d$ with the positive real
line.  (We note in passing that this quantity, known as the \emph{uniqueness
  threshold} for the hard core model, has played a central role in the study of
the algorithmic estimation of $Z_G$ on the positive real line: in particular,
this study led to some of the tightest known connections between statistical
mechanics phase transitions and computational
complexity~\cite{Weitz,sly12,gsv,Sly2010CompTransition}.)  Peters and
Regts~\cite{peters_conjecture_2019} showed that for any positive
$\lambda' < \lambda_c(d)$, there is an $\epsilon' = \epsilon'(\lambda') > 0$
such that for any $z$ satisfying $\abs{\Im z} \leq \epsilon'$ and
$\Re z = \lambda'$, $Z_G(z) \neq 0$ for all $G \in \cG_{d+1}$.  They also gave
explicit lower bounds on $\epsilon'(\lambda')$ for
$\lambda' \in (0, \tan (\pi/(2d)))$ (note that $\tan (\pi/(2d)) > \lambda^*(d)$
for $d \geq 2$, so these results are not implied by \Cref{thm:shearer}).  Bencs
and Csikvári~\cite{bencs18:_note}, using different methods, improved on the
latter lower bounds, and thereby significantly extended the known zero-free
region inside $U_d$ in the right half-plane.  In \cref{sec:related-work} below,
we describe some more recent papers that study phenomena such as the limit shape
(after appropriate scaling) of the zero-free region as $d \uparrow \infty$, and
that explore further connections between zero-freeness and other aspects of the
independent set model.  However, for specific finite $d$, none of these results
seem to provide any new zero-free regions in the left half-plane beyond the
half-disk implied by \Cref{thm:shearer}.

\subsection{Contributions} In this paper, we give two geometric criterions
(\Cref{thm:Simons-result,thm:init-curve}) which together give a framework for
rigorously establishing  (connected) zero-free regions as well as a way to carry out
semi-rigorous numerical explorations. 

We provide several examples of the (rigorous) use of these criterions.  We
establish two new zero-free regions in the left half plane:
\Cref{thm:critical-vicinity} gives a better result in the vicinity of the
negative real line, while \Cref{thm:lhp} gives a better result near the
imaginary line.  When restricted to the imaginary axis, the latter region agrees
with the result of Bencs and Csikvári~\cite{bencs18:_note} for the right
half-plane.  We also extend the previous zero-freeness results of Bencs and
Csikvári~\cite{bencs18:_note} for the right-half plane using the geometric
criterions developed in this paper (\cref{lem:righthalfplane_better}, see also
\cref{rem:right-improvement}).  We also show that our framework gives a new
proof of the Sokal conjecture, which was first proved by Peters and Regts via a
potential function argument~\cite{peters_conjecture_2019}
(\Cref{thm:sokal-conj}).  See \cref{fig:d9} for a graphical illustration of
these new zero-freeness results.

\paragraph{Algorithmic implications} Following the template provided by the results
of Barvinok~\cite{barvinok2017combinatorics} and Patel and
Regts~\cite{patel2017deterministic}, these new zero-freeness results also
immediately lead to new polynomial-time algorithms for the approximation of
$Z_G(\lambda)$ for $\lambda$ lying in the interior of these regions on graphs
$G \in \cG_{d+1}$ of maximum degree at most $d + 1$. Given complex numbers $Z$ and $\hat{Z}$ we say that $\hat{Z}$ is a multiplicative  $\varepsilon$-approximation of $Z$ if $e^{-\varepsilon}<\frac{|\hat{Z}|}{|Z|}<e^{-\varepsilon}$ and the angle between $\hat{Z}$ and $Z$
 considered as vectors in $\mathbb{C}=\mathbb{R}^2$ is at most $\varepsilon$. Then the algorithmic framework of Barvinok~\cite{barvinok2017combinatorics} and Patel and
Regts~\cite{patel2017deterministic} combined with our zero-free regions provides a deterministic algorithm  of running time  $\left(\frac{|V|}{\varepsilon}\right)^{O_{d,\lambda}(1)}$ for obtaining a multiplicative $\varepsilon$-approximation of $Z_G(\lambda)$ whenever $G=(V,E)$ has maximum degree $d+1$ and $\lambda$ is in the interior of the zero-free region provided by this paper.

\paragraph{Numerical explorations} We now comment briefly on the connections to
numerical explorations -- alluded to above -- of our work.  The naive method to
numerically check whether a point $\lambda$ is in the zero free region of $Z_G$
for all graphs in $\cG_{d+1}$ would be to evaluate $Z_G(\lambda)$ for all such
graphs, and to check if it evaluates to $0$.  Known results allow one to
restrict the set of graphs one has to explore to trees in $\cG_{d+1}$ (see
\cref{original} below), but the resulting procedure is still computationally
infeasible.  In contrast, the geometric criterions in
\cref{thm:Simons-result,thm:init-curve} allow one to do the following. Given $d$
and $\lambda$, one tries to construct a curve in the complex plane with certain
prescribed properties.  The existence of such a curve then certifies that no
graph in $\cG_{d+1}$ has $Z_G(\lambda) = 0$.  What curves would ``work'' for a
given $d$ and $\lambda$ can then be explored numerically: in fact, many of our
zero-freeness results listed above were obtained by first conjecturing the form
of such a curve guided by numerical experiments, and then rigorously verifying
-- as done in the proofs of the theorems listed above -- that the curve has the
prescribed properties.  We want to highlight, however, that although this method
is ``sound'' -- in the sense that producing such a curve as a certificate
guarantees zero-freeness -- it is not necessarily ``complete'' -- one may
not be able to construct such a curve as a certificate even though $\lambda$ is
in the zero-free region.

 \begin{figure}[!ht]
   \centering \includegraphics[width=0.8\textwidth]{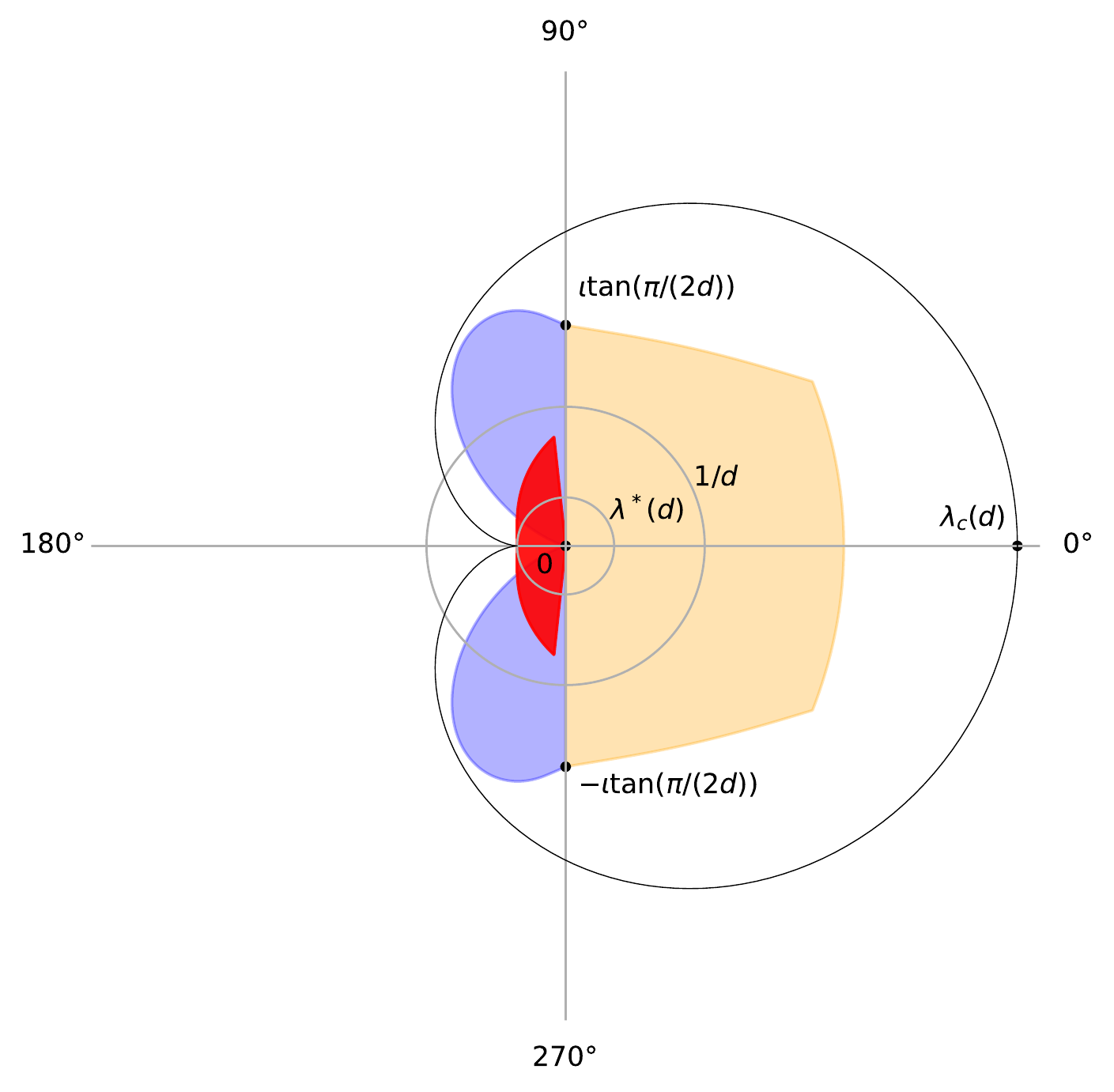}
   \caption{New zero free regions for $d=9$ (graphs of degree at most $10$). In
     the left half plane, the red region corresponds to
     \cref{thm:critical-vicinity}, while the blue region corresponds to
     \cref{thm:lhp}. In the right half plane, the yellow region corresponds to
     \cref{lem:righthalfplane_better}. The smaller grey circle around the origin
     has radius $\lambda^*(d)$ (the ``Shearer radius'' from \cref{thm:shearer}),
     and the points $\pm\iota\tan(\pi/(2d))$ are marked on the imaginary axis.
     The outer black ``cardioid-shaped'' curve is the boundary $\partial U_d$ as
     defined in \cref{eq:12}.  A magnified version of the red region
     (corresponding to \cref{thm:critical-vicinity}) is given in \Cref{fig:1}.}
     \label{fig:d9}
 \end{figure}
 
 \subsection{Organization of the paper} After a short section of preliminaries,
 we introduce in \cref{relaxed_recurrence} various simple criterions to prove
 that a $\lambda \in \bC$ is in the zero-free region of independence polynomial
 of graphs of bounded degree. Building on this work, we introduce in
 \cref{sec:crit-orig-compl} two new criterions
 (\Cref{thm:Simons-result,thm:init-curve}) which use constructions of certain
 curves in order to prove zero-free regions.  The remaining sections are direct
 applications of these two criterions, and are independent of each other. Since
 some of the proofs are somewhat technical, these sections are arranged in
 increasing order of difficulty. In \cref{sec:sokal-conjecture} we give a
 new proof of Sokal's conjecture originally proven by Peters and Regts.  In
 \cref{sec:zero-free-region-negative-real-line} we give a new zero-free
 region in the vicinity of the critical point $\frac{d^d}{(d+1)^{d+1}}$. In
 \cref{sec:zero-free-region-imaginary-axis} we provide a zero-free region
 close the imaginary axis. Finally, in \cref{sec:right-half-plane} we
 prove a zero-free region in the right half plane.
 
 For those readers who are interested in the ideas in general, but want to avoid
 technical difficulties we recommend reading the paper  till the end of
 \cref{sec:sokal-conjecture} and omitting \cref{thm:Simons-result}
 and its proof.

 \subsection{Related work}
 \label{sec:related-work}
 The work of Barvinok~\cite{barvinok_computing_2015} (see also
 \cite{barvinok2017combinatorics}) pioneered the direct use of zero-free regions
 for designing algorithms for approximate counting.  However, in most examples,
 a direct application of Barvinok's method gives a quasi-polynomial time
 algorithm: Patel and Regts~\cite{patel2017deterministic} showed how to use
 various combinatorial tools in order to reduce this quasi-polynomial runtime to
 a polynomial runtime in various ``bounded-degree'' settings.  The method has
 since then been used to attack a wide variety of approximate counting problems:
 see, e.g., \cite{barvinok_computing_2015,
   barvinok15:_comput_perman_some_compl_matric, BarvinokSoberon16a,
   eldar2018approximating, harrow_classical_2020,
   liu19:_deter_algor_count_color_delta_color, bencs2018zero}.  How this method
 relates to other methods of approximate counting, such as Markov chain Monte
 Carlo, or the method of reduction to tree-recurrences and ``correlation decay''
 (first used by Bandyopadhyay and Gamarnik~\cite{bandyopadhyay_counting_2008}
 and Weitz~\cite{Weitz}), has also been explored in several papers, see, e.g.,
 \cite{shaosun19, JingchengThesis, li_complex_2021, liu2018fisher}.  In the
 statistical physics literature, very strong connections between Markov chain
 Monte Carlo and zero-freeness are known in the special case of integer lattices
 through the work of Dobrushin and Shlosman~\cite{DS85,DS87}.

 As already discussed in detail in the previous subsections, the complex zeros
 of the independence polynomial have also been studied extensively in the
 context of its connections to the Lovász local lemma and also in the context of
 its computational complexity~\cite{shearer_problem_1985, scott_repulsive_2005,
   patel2017deterministic, 10.5555/3174304.3175407,
   bezakova_inapproximability_2018}.  Here we describe a few more recent works
 in this direction.  Recent work of de Boer, Buys, Guerini, Peters, and
 Regts~\cite{de2021zeros} establishes strong formal connections between the
 computational complexity of the hard core model, complex dynamics, and
 zero-freeness of the partition function (see Main Theorem
 of~\cite{de2021zeros}): in particular they prove that the zeros of $Z_G$ for
 graphs in $\cG_{d+1}$ are dense in the complement of $U_d$.  The limit shapes
 of the zero-free regions have also been studied: Bencs, Buys, and
 Peters~\cite{bencs2021limit} show that in the $d\to \infty$ limit, a rescaled
 version of zero-free region tends to a bounded 0-star shaped region, whose
 boundary intersects $\lim_{d\to \infty}d\cdot \partial U_d$ only at real
 parameters. In particular, the results of~\cite{bencs2021limit} show that for
 large enough $d$ the zero-free region is strictly contained in $U_d$ except for
 the two real parameters.  The problem of fully characterizing the zero-free
 region of $Z_G$ for $G\in\cG_{d+1}$ in the complex plane, however, still
 remains open.

\section{Preliminaries}

\paragraph*{Branch cuts}
We adopt the following convention for defining fractional powers and complex
logarithms.  Given $z = re^{\iota \theta}$ with $r > 0$ and
$\theta \in (-\pi, \pi]$, we define
\begin{align}
  \log z &\defeq \log r + \iota \theta\text{, and}\label{eq:5}\\
  z^\delta &\defeq r^{\delta}\exp(\iota \delta \theta)\text{, for any $\delta > 0$}.\label{eq:6}
\end{align}
We leave the functions undefined when $z = 0$ (except that we adopt the usual
convention that $0^0 = 1$).  Note that with the above definition, $\log$ and
$z^\delta$ for non-integral $\delta$ are defined but discontinuous on the
negative real line.  However, we do have the following identity for all $z \neq
0$ and $\delta \geq 0$:
\begin{displaymath}
  z^\delta = \exp(\delta \log z).
\end{displaymath}
Further, for $z \neq 0$, we use the convention $\arg z = \Im (\log z)$.

\paragraph*{Graphs and independence polynomials}
For the sake of providing a quick reference, we recollect here some basic
notation and terminology about graphs and their independence polynomials that
was introduced in the introduction above.  We denote the set of all graphs of
degree at most $d + 1$ by $\cG_{d+1}$.  The \emph{independence polynomial}
$Z_G(\lambda)$ of a graph $G$ is given by
\begin{equation}
  Z_G(\lambda) \defeq \sum_{I:\text{independent set in $G$}}\lambda^{|I|}.\label{eq:4}
\end{equation}
Two quantities of interest with respect to the independence polynomial are the
\emph{Shearer radius} $\lambda^*(d) \defeq \frac{d^d}{(d + 1)^{d+1}}$, and the
\emph{uniqueness threshold} $\lambda_c(d) \defeq \frac{d^d}{(d - 1)^{d + 1}}$.
The former, $\lambda^*(d)$, is specially connected to the Lov\'{a}sz local
lemma, and also the radius of the largest circular disk around the origin in
which $Z_G$ is zero-free for all graphs in $\cG_{d+1}$ (see \cref{thm:shearer}
above).  The latter, $\lambda_c(d)$, is intimately connected to the complexity
of approximating $Z_G(\lambda)$ for $G \in \cG_{d+1}$, for $\lambda$ on the
positive real line: in particular, Weitz~\cite{Weitz} gave a deterministic fully
polynomial approximation scheme (FPTAS) for $Z_G(\lambda)$ for
$G \in \cG_{d+1}$, provided $\lambda < \lambda_c(d)$, while in a series of
works~\cite{Sly2010CompTransition,sly12,Vigoda-hard-core-11,gsv} starting with a
paper of Sly, it was shown that a randomized fully polynomial approximation
scheme (FPRAS) for the same problem in the regime $\lambda > \lambda_c(d)$ would
imply NP = RP.  Further, in terms of the curve $\partial{}U_d$ of Peters and
Regts~\cite{peters_conjecture_2019} (see \cref{eq:12} above), $\lambda_c(d)$ is the
unique point of intersection of $\partial{}U_d$ with the positive real line,
while $-\lambda^*(d)$ is the unique point of intersection of $\partial{}U_d$
with the negative real line.

\section{A relaxed recurrence for the independence polynomial} \label{relaxed_recurrence}

As stated in the introduction, our goal is to study the location of the complex
zeros of the (univariate) independence polynomial
$$ Z_G(\lambda) \defeq \sum_{I \subseteq V(G) \mbox{ independent}} \lambda^{|I|}.$$
Recall that we focus on the class of graphs $\cG_\Delta$ with degrees at most
$\Delta$ for some fixed $\Delta \geq 3$.  It is more convenient, however, to
work in terms of the notation $d \defeq \Delta-1$.  We now proceed to describe a
known characterization of zero-free regions for the independence polynomial of
bounded degree graphs, in preparation for which we introduce the following
definition.

\begin{definition}[\textbf{The set $S_\lambda = S_\lambda(d)$}]
  For $\lambda \in \CC$, define $S_\lambda \subseteq \CC$ as the set of points
  that can be generated by the following rules:
\begin{itemize}
\item $0 \in S_\lambda(d)$,
\item If $z_1,\ldots,z_d \in S_\lambda(d)$ are such that $z_i \neq -1$ for
  $1\leq i \leq d$, then
\begin{equation}
\label{eq:basic-recurr}
f(z_1,\ldots,z_d) = \frac{\lambda}{\prod_{i=1}^{d} (1+z_i)}
\end{equation}
  is also in $S_\lambda(d)$.
\end{itemize}
\end{definition}
(Although the definition of $S_\lambda$ depends on $d$, we will often omit this
dependence from our notation when the value of $d$ is clear from the context.)

The following theorem is well known~\cite{Weitz,scott_repulsive_2005},
and has been used in previous work on the subject
(e.g. in~\cite{peters_conjecture_2019,bencs18:_note}). It can most directly be
obtained from a result of Bencs~\cite{bencs_trees_2018}, who showed that the
independence polynomial of a graph divides (as a polynomial) the independence
polynomial of the so-called ``self-avoiding walk tree'' of the graph.  The
zero-free regions of the independence polynomial of a tree can in turn be
analyzed in terms of the ``tree recurrences'' described in
\cref{eq:basic-recurr}~\cite{Weitz,scott_repulsive_2005}.
\begin{theorem}[\textbf{see, e.g., Proposition 2.7 (1) of
    \cite{bencs_trees_2018}, and Lemma 2.1 of
    \cite{bencs18:_note}}] \label{original} Fix $d \geq 2$.  $Z_G(\lambda) = 0$
  for some graph $G \in \cG_{d+1}$ if and only if $-1 \in S_\lambda(d)$.
\end{theorem}

A standard application of \cref{original} is to define a ``trapping
region'' $T$ such that $0\in T$, $-1\notin T$ and
$f$ maps $T$ to $T$. For instance, if
$|\lambda|\leq \frac{d^d}{(d+1)^{d+1}}$, then
$T=\{z \in \CC\ |\ |z|\leq \frac{1}{d+1}\}$ is such a region:
$$\left|\frac{\lambda}{\prod_{i=1}^d(1+z_i)}\right|\leq \frac{d^d}{(d+1)^{d+1}}\prod_{i=1}^d\frac{1}{1-\frac{1}{d+1}}=\frac{1}{d+1}$$
showing Shearer's result.  In general, it is not easy to handle $d$ variables at
the same time. Therefore, in what follows we try to find sufficient conditions
that only require understanding the behaviour of a univariate map.

In the following, we relax the recurrence $f(z_1,\ldots,z_d)$ to allow fractional powers and more than $d$ arguments. As we will see, this in fact leads to a simplification of the problem.

\begin{definition}[\textbf{The set $\tilde{S}_\lambda = \tilde{S}_\lambda(d)$}]
  For $\lambda \in \CC$, define $\tilde{S}_\lambda \subseteq \CC$ as the set of
  points that can be generated by the following rules:
\begin{itemize}
\item $0 \in \tilde{S}_\lambda(d)$,
\item If $z_1,\ldots,z_k \in \tilde{S}_\lambda(d)$ and
  $\delta_1,\ldots,\delta_k \geq 0$ are such that
  $\sum_{i=1}^{k} \delta_i \leq d$ and $z_i \neq -1$ for $1 \leq i \leq k$, then
\begin{equation}
\label{eq:ext-recurr}
 f_{\delta_1,\ldots,\delta_k}(z_1,\ldots,z_k) = \frac{\lambda}{\prod_{i=1}^{k} (1+z_i)^{\delta_i}}
\end{equation}
  is also in $\tilde{S}_\lambda(d)$.
\end{itemize}
\end{definition}
(As with $S_\lambda$, although the definition of $S_\lambda$ depends on $d$, we
will often omit this dependence from our notation when the value of $d$ is clear
from the context.)

Clearly, we have $S_\lambda(d) \subseteq \tilde{S}_\lambda(d)$, since the new
generation rule subsumes $f(z_1,\ldots,z_d)$. Hence, from \cref{original}, we
directly obtain the following.

\begin{lemma}
  Fix $d \geq 2$.  If $-1 \notin \tilde{S}_\lambda(d)$, then
  $Z_G(\lambda) \neq 0$ for every $G \in \cG_{d+1}$.
\end{lemma}

The main advantage of the relaxed recurrence is that it allows us to replace the
multivariate recurrence by a univariate one. We do this as follows: Consider the
set $\{ \log (1+z): z \in \tilde{S}_\lambda \}.$ Note that this is well defined
if $-1 \not \in \tilde{S}_\lambda$. If we write $w_i = \log (1+z_i)$, then the
recurrence
$$f_{\delta_1,\ldots,\delta_k}(z_1,\ldots,z_k) = \frac{\lambda}{\prod_{i=1}^{k} (1+z_i)^{\delta_i}} $$ can be rewritten by substitution as
$$ g_{\delta_1,\ldots,\delta_k}(w_1,\ldots,w_k) = \log (1 + f_{\delta_1,\ldots,\delta_k}(e^{w_1}-1,\ldots,e^{w_k}-1)) $$
 $$ = \log \left(1 + \lambda \prod_{i=1}^{k} e^{-\delta_i w_i} \right) = \log \left( 1 + \lambda e^{-\sum_{i=1}^{k} \delta_i w_i} \right).$$

Hence, a combination of fractional powers in $f_{\delta_1,\ldots,\delta_k}$ corresponds to a linear combination of the points $w_i = \log (1+z_i)$.
If we normalize the linear combination by $\frac{1}{d}$, and use the fact that
$0$ is always a possible choice for $w_i$, we obtain a convex linear combination
of $w_1,\ldots,w_k$ in the exponent. This motivates the following
characterization.  Note that the characterization is in terms of the behavior of
a function of only one complex variable.

\begin{theorem}
\label{thm:w-criterion}
Fix $d \geq 2$.  The number $-1$ is not contained in $\tilde{S}_\lambda(d)$ if
and only if there is a convex set $T \subset \CC$ containing $0$ such that for
every $w \in T$,
$$ g(w) = \log (1 + \lambda e^{-dw}) $$
is well-defined and $g(w) \in T$.
\end{theorem}

\begin{proof}
  Suppose first that $-1 \not\in \tilde{S}_\lambda = \tilde{S}_\lambda(d)$.  We
  define
  \begin{displaymath}
    T = \conv{\log(1+z) | z \in \tilde{S}_\lambda}.
  \end{displaymath}
  Note that since $-1 \not\in \tilde{S}_\lambda$, $T$ is well-defined, and
  further, is convex by definition.  Also, $0 \in T$, since
  $0 \in \tilde{S}_\lambda$.  Now consider $w \in T$.  By Caratheodory's
  theorem, there exist $\delta_1, \delta_2, \delta_3 \geq 0$ summing up to $d$,
  and $z_1, z_2, z_3 \in \tilde{S}_\lambda$, such that
  $w = \frac{1}{d}\sum_{i=1}^3\delta_i\log(1+z_i)$.  We thus have
  $\lambda\exp(-dw) = f_{\delta_1, \delta_2, \delta_3}(z_1, z_2, z_3) \in
  \tilde{S}_\lambda$.  Thus, $\lambda\exp(-dw) \neq -1$ and hence
  $g(w) = \log (1 + \lambda\exp(-dw))=\log(1+f_{\delta_1, \delta_2, \delta_3}(z_1, z_2, z_3))$ is well-defined and lies in $T$.

  Conversely, suppose that $T$ is any arbitrary convex set containing $0$, on
  which the map $g(w) = \log (1 + \lambda \exp(-dw))$ is well defined, and
  satisfies $g(w) \in T$ for all $w \in T$.  We claim that if
  $-1 \in \tilde{S}_\lambda$, then there exists $w \in T$ such that
  $-1 = \lambda\exp(-dw)$.

  To see this, define the \emph{depth} of every
  $z\in \tilde{S}_\lambda$ as follows: $\mathrm{depth}(0) = 0$, and for
  $z \neq 0$, $\mathrm{depth}(z)$ is the smallest integer $D$ such that $z$ can
  be written as $f_{\delta_1, \delta_2, \dots, \delta_k}(z_1, z_2, \dots, z_k)$
  where $k$ is a positive integer, $\delta_i \geq 0$ sum to at most $d$, and
  $z_i \in \tilde{S}_\lambda$ have depth at most $D-1$.  Note that
  $\mathrm{depth}(z) \geq 1$ for $z \neq 0$.  Now, if
  $-1 \in \tilde{S}(\lambda)$, let $D_{-1} = \mathrm{depth}(-1)$.

  We claim now that for all $z \in \tilde{S}_\lambda$ of depth at most $D_{-1} - 1$,
  $\log(1+z) \in T$.  This is proved by induction on the depth of $z$: it is
  true in the base case $\mathrm{depth}(z) = 0$ (so that $z = 0$), since
  $0 \in T$.  Otherwise, from the definition of depth, we can find
  $z_1, z_2, \dots z_k$ of depth strictly smaller than $z$, and
  $\delta_i \geq 0$ summing up to at most $d$, such that
  \begin{equation}
    z = f_{\delta_1, \delta_2, \dots, \delta_k}(z_1, z_2, \dots, z_k) = \lambda\exp\inp{-d\sum_{i=1}^k\frac{\delta_i}{d}\log(1+z_i)}.\label{eq:7}
  \end{equation}
  Thus, we have $\log(1 + z) = g(w)$ where $w$ is a convex combination of $0$
  and the quantities $\log(1 + z_i)$.  The latter quantities are all inductively
  in $T$, so that $w$ is also in $T$ (as $T$ is convex).  But since $T$ is
  closed under applications of $g$, this implies that $g(w) = \log(1+z)$ is also
  in $T$.  This establishes the claim that for every $z \in \tilde{S}_\lambda$
  of depth at most $D_1 - 1$, $\log(1 + z)$ is an element of $T$.

  Now, applying the argument leading to \cref{eq:7} with $z = -1$ (which by
  assumption has depth $D_{-1}$), we conclude that there exists a $w \in T$ such
  that $-1 = \lambda\exp(-dw)$.  But this contradicts the hypothesis that
  $g(w) = \log(1 + \lambda\exp(-dw))$ is well-defined on $T$.  Thus, it cannot
  be the case that $-1 \in \tilde{S}_\lambda$.
\end{proof}

For natural reasons, we call a $T$ as in the statement of the above theorem a
{\em trapping region} for $\lambda$.

\

{\em Remark:} Sometimes it is desirable to avoid $-1$ even in the closure of $S_\lambda$, or $\tilde{S}_\lambda$. By our transformation, this corresponds to the property that there is a convex set $T$ containing $0$ and closed under $g(w) = \log (1 + \lambda e^{-dw})$, such that $\Re(w) \geq -K$ for every $w \in T$ and some constant $K>0$. This is equivalent to saying that every point $z \in \tilde{S}_\lambda$ satisfies $|1+z| \geq e^{-K}$.

\section{Criterion in the original complex plane}
\label{sec:crit-orig-compl}

The previous section shows that we get a rather clean picture when we study the behavior of the extended recurrence (\cref{eq:ext-recurr}) after a change of variable, $w = \log (1+z)$.
However, we can also formulate a criterion using trapping regions in the original variable $z$. This criterion looks more intuitive, but it seems we lose a bit in the transition (in particular, we do not get an equivalence here).

\begin{theorem}
\label{thm:z-criterion}
If there is a convex set $S \subset \CC$ containing $0$, not containing $-1$, such that $f(z) = \frac{\lambda}{(1+z)^d} \in S$ for every $z \in S$, then $Z_G(\lambda) \neq 0$ for every $G \in \cG_\Delta$.
\end{theorem}

In order to prove this statement, we need the following fact about the behavior
of arithmetic vs.~geometric averages in the complex plane.  While we believe
this fact to be standard, we are unable to find an exact reference, and hence
provide a proof for completeness.\footnote{Note that, despite the title, the
  lemma does not contradict the usual inequality between the arithmetic and the
  geometric means of positive reals.  The lemma is in fact a trivial statement
  for the case of positive reals.}

\begin{lemma}[``geometric averages dominate arithmetic averages'']
\label{lem:AMGM}
For any two points $z_1, z_2 \in \CC \setminus \{0\}$ satisfying
$\abs{\arg(z_1) - \arg (z_2)} \leq \pi$, and $\alpha \in [0,1]$, there exist
$\beta \in [0,1]$ and $t \in [0,1]$ such that
$$ t z_1^\alpha z_2^{1-\alpha} = \beta z_1 + (1-\beta) z_2.$$
\end{lemma}

\begin{proof}
  We reduce to the case where $z_2 = 1$, by dividing by $z_2$ and substituting
  $z = z_1/z_2$.  Our goal then is to find $t, \beta \in [0,1]$ such that
  $$ t z^\alpha = \beta z + (1-\beta).$$
  We can also assume that $\arg z \in [0, \pi]$, by complex conjugation if this
  is not the case.

  Now, if $\arg z = 0$ the claim is trivially true because, then, if
  $z = \Re z \geq 1$, we can take $\beta = 0, t = z^{-\alpha}$, while when
  $0 < z = \Re z < 1$, we can take $\beta = 1, t = z^{1-\alpha}$.  Similarly,
  when $\arg z = \pi$, $z$ is a strictly negative real number, so that we can
  choose $t = 0$ and $\beta = \frac{1}{1-z} \in [0, 1]$.

  We can thus assume that $ \theta \defeq \arg(z) \in (0,\pi)$ and
  $r \defeq \abs{z} > 0$.  Note that $\arg(z^\alpha) = \alpha\theta$.  Let
  $y(\alpha)$ be the unique point with argument $\alpha\theta$ on the line
  segment joining $1$ and $z$. From an elementary geometric argument, we then
  have
\begin{displaymath}
  \abs{y(\alpha)} = \frac{r \sin \theta}{r\sin((1-\alpha)\theta) + \sin \alpha\theta}.
\end{displaymath}
We now define the function
$f(x): [0, 1] \rightarrow \R$ as
\[f(x) \defeq \log \frac{\abs{y(x)}}{\abs{z^x}} = \log \frac{r^{1-x}\sin
    \theta}{r\sin((1-x)\theta) + \sin x\theta}. \] Note that the claim of the
lemma is equivalent to showing that $f(x) \leq 0$ for all $x \in [0, 1]$ (the
quantity $t$ can then be taken to $e^{f(\alpha)} \in [0,1]$ and
$\beta \in [0,1]$ is such that $\beta z + 1 - \beta = y(\alpha)$).

To this end, we first note that $f(0) = f(1) = 0$, so the claim would follow if
$f$ is convex on $[0, 1]$.  We verify this by directly computing the second
derivative of $f$ and checking that it is non-negative in $[0, 1]$:
\begin{displaymath}
  f''(x) = \frac{\theta^2 (1 + r^2 - 2r \cos \theta)}{\inp{r\sin((1-x)\theta) +
      \sin x\theta}^2} \geq 0, \text{ when $x \in [0, 1]$}.\qedhere
\end{displaymath}
\end{proof}

\begin{proof}[Proof of \cref{thm:z-criterion}]
Suppose that there is a convex set $S$ containing $0$, not containing $-1$, and closed under the map $f(z) = \frac{\lambda}{(1+z)^d}$.
We will transform $S$ into a convex set $T$ satisfying the assumptions of \cref{thm:w-criterion}. Define
$$ T = \conv{\log (1+z): z \in S}.$$
By construction, $T$ is convex and it contains $0$ (since $0 \in S$). We need to prove that $g(w) = \log (1+\lambda e^{-dw})$ is well-defined on $T$ and preserves membership in $T$.
Consider $w \in T$, i.e., $w = \sum_{i=1}^{k} \alpha_i \log (1+z_i)$, a convex combination of points $\log (1+z_i)$ such that $z_i \in S$.
We can set $\delta_i = d \alpha_i$, hence $w = \frac{1}{d} \sum_{i=1}^{k} \delta_i \log (1+z_i)$. Then,
$$ g(w) = \log (1 + \lambda e^{-dw}) = \log \left(1 + \frac{\lambda}{\prod_{i=1}^{k} (1+z_i)^{\delta_i}} \right).$$
Now, we appeal to \cref{lem:AMGM}. Note first that
$1 + S \defeq \inb{1 + z \vert z \in S}$ is a convex set containing $1$ and not
containing $0$, so that by the separating hyperplane theorem, all of $1 + S$
lies in a halfplane defined by a line passing through $0$, and therefore
$\abs{\arg(1 + u) - \arg(1 + v)} \leq \pi$ is true for all $u, v \in S$. Now, we
take one pair $z_i, z_j$ of points at a time, and consider
$u_{ij} = (1+z_i)^{\frac{\delta_i}{\delta_i+\delta_j}}
(1+z_j)^{\frac{\delta_j}{\delta_i+\delta_j}}$. By \cref{lem:AMGM}, there is
$\beta \in [0,1]$ such that $1+z_{ij} = \beta (1+z_i) + (1-\beta) (1+z_j)$ is a
point of the same argument and smaller-or-equal modulus as $u_{ij}$. Hence we
can replace both $z_i$ and $z_j$ by $z_{ij}$ and continue. We maintain the
property that the argument of $\prod (1+z_i)^{\delta_i}$ remains preserved and
the modulus can only decrease. Eventually, we obtain a point
$\tilde{z} \in \conv{z_1,\ldots,z_k} \subseteq S$ such that
$\arg((1+\tilde{z})^d) = \arg(\prod (1+z_i))$ and
$|(1+\tilde{z})^d| \leq | \prod (1+z_i) |$. Hence, we can write
$$ g(w) = \log \left( 1 + \frac{c\cdot \lambda}{(1+\tilde{z})^d} \right),$$
where $c \leq 1$ is a non-negative real number.  Now, note that since
$\tilde{z} \in S$, we have
$f(\tilde{z}) = \frac{\lambda}{(1 + \tilde{z})^d} \in S$, as $S$ is closed under
applications of $f$.  Then, since $0 \in S$, and $S$ is convex, we get
$y \defeq \frac{c\cdot \lambda}{(1+\tilde{z})^d} \in S$, and further that
$y \neq -1$, as $-1 \not\in S$.  Thus, by definition of $T$,
$g(w) = \log (1 + y) \in T$, as $y \in S$.  This proves that $T$ satisfies the
assumptions of \cref{thm:w-criterion} and hence $-1$ is not contained in
$\tilde{S}_\lambda$, which implies that $Z_G(\lambda) \neq 0$.
\end{proof}

Next, we present a more abstract extended version of this criterion, where we
allow a ``convex'' initial segment $h(t), t \in [0,1]$ rather than a line
segment. First we define the following notion:

\begin{definition}[\textbf{$-1$-covered points}]
\label{def:-1-cover}
A point $z \in \CC$ is \emph{$-1$-covered} by $z' \in \CC$ if
$\arg(1+z) = \arg(1+z')$ and $|1+z| \geq |1+z'|$.  More generally, a set $T$ is
$-1$ covered by a set $S$ if for every $z \in T$, there is a point $z' \in S$
such that $z$ is $-1$-covered by $z'$.
\end{definition}
\noindent Geometrically, the above notion captures $z$ being ``covered'' by $z'$
when ``viewed'' from the point $-1$.  The utility of this definition for our
purposes comes from the following simple observation.

\begin{observation}
  \label{obv:covering}
  Fix an integer $d \geq 2$ and a $\lambda \in \bC$, and consider
  $f(z) \defeq \frac{\lambda}{(1 + z)^d}$.  If $z\in \bC$ is $-1$-covered by
  $w \in \bC$, then $f(z) = \alpha f(w)$ for some $\alpha \in [0, 1]$.
\end{observation}
\begin{proof}
  If $\lambda = 0$, there is nothing to prove, so assume $\lambda \neq 0$.
  Since $z$ is $-1$-covered by $w$, we have $\abs{1 + z} \geq \abs{1 + w}$ and
  $\arg (1+w) = \arg(1+z)$.  It follows that $\arg(f(z)) = \arg(f(w))$, and
  $\abs{f(z)} \leq \abs{f(w)}$.  Thus, $f(z)$ lies on the segment joining the
  origin to $f(w)$, and the claim follows.
\end{proof}

Next we state and prove our first main geometric criterion for zero-freeness,
which will be applied multiple times in the subsequent sections.

\begin{theorem}
\label{thm:init-curve}
For $\lambda \in \CC$, assume that there is a curve $\{ h(t): t \in [a,b]\}$,
where $a < b$ are real numbers, such that
\begin{itemize}
\item $h(t) = 0$ for some $t \in [a,b]$,
\item $\arg(1+h(t))$ is strictly increasing for $t \in [a,b)$,
\item $h(t)$ is ``convex'' in the sense that for any $t_1,t_2 \in [a,b], \alpha \in [0,1]$, $\alpha h(t_1) + (1-\alpha) h(t_2)$ is $-1$-covered by $h(t)$ for some $t \in [t_1,t_2]$.
\item for every $t \in [a,b]$, $f(h(t)) \defeq \frac{\lambda}{(1+h(t))^d}$ is $-1$-covered by $h(t')$ for some $t' \in [a,b]$.
\end{itemize}
Then $Z_G(\lambda) \neq 0$ for any $G \in \cG_{d+1}$.
\end{theorem}

\begin{proof}
  By an affine reparameterization of the curve $h$, if necessary, we assume that
  $a = 0$ and $b = 1$.  Given the curve $h(t), t \in [0,1]$, we define a
  trapping region, in the sense of \Cref{thm:z-criterion}, as a ``shadow of the
  curve $h(t)$ when illuminated from the point $-1$'':
$$ S = \{z \in \CC: \exists t \in [0,1] \textrm{ such that } z \mbox{ is $-1$-covered by } h(t) \}.$$
This is a convex set, since for any $z_1, z_2 \in S$, $z_1$ and $z_2$ are
$-1$-covered by $h(t_1), h(t_2)$ respectively, $z = \alpha z_1 + (1-\alpha) z_2$
is $-1$-covered by $z'= \alpha' h(t_1) + (1-\alpha') h(t_2)$ (for some
$\alpha' \in [0, 1]$), and $z'$ in turn is covered by $h(t)$ for some
$t \in [t_1,t_2]$ by the convexity of $h(t)$.

Also, $0$ is contained in $S$ because $h(t) = 0$ for some $t \in [0,1]$; $-1$
is not contained in $S$, since $0$ is the only real value on the curve $h(t)$
(this follows since $\arg (1 + h(t))$ is assumed to be a strictly increasing
function of $t$), so that $S$ contains only non-negative real numbers.

To apply \Cref{thm:z-criterion} in order to conclude the proof, it remains to
prove that $S$ is closed under the map $f(z) = \frac{\lambda}{(1+z)^d}$. For any
$z \in S$, there is a $t \in [0,1]$ such that $z$ is $-1$-covered by
$h(t)$. \Cref{obv:covering} then implies that $f(z) = \alpha f(h(t))$ for some
$\alpha \in [0, 1]$.  By the assumptions of the theorem, we know that $f(h(t))$
is $-1$-covered by some point $h(t')$, $t' \in [0,1]$, which implies that
$f(h(t)) \in S$. Hence by convexity, since $0 \in S$, we also get
$f(z) = \alpha f(h(t)) \in S$, as required.
\end{proof}

Next, we formulate a more concrete sufficient condition which can be used in
numerical experiments.\footnote{This was the result stated
    in a talk at the Simons Institute for the Theory of Computing, UC Berkeley,
    on March 18, 2019.} The proof we give here highlights the connections of
this result with numerical exploration, even though alternative proofs may be
possible (see \cref{rem:numerical,rem:alternative} following the proof).
\begin{figure}[t]
  \centering
  \includegraphics[scale=0.8]{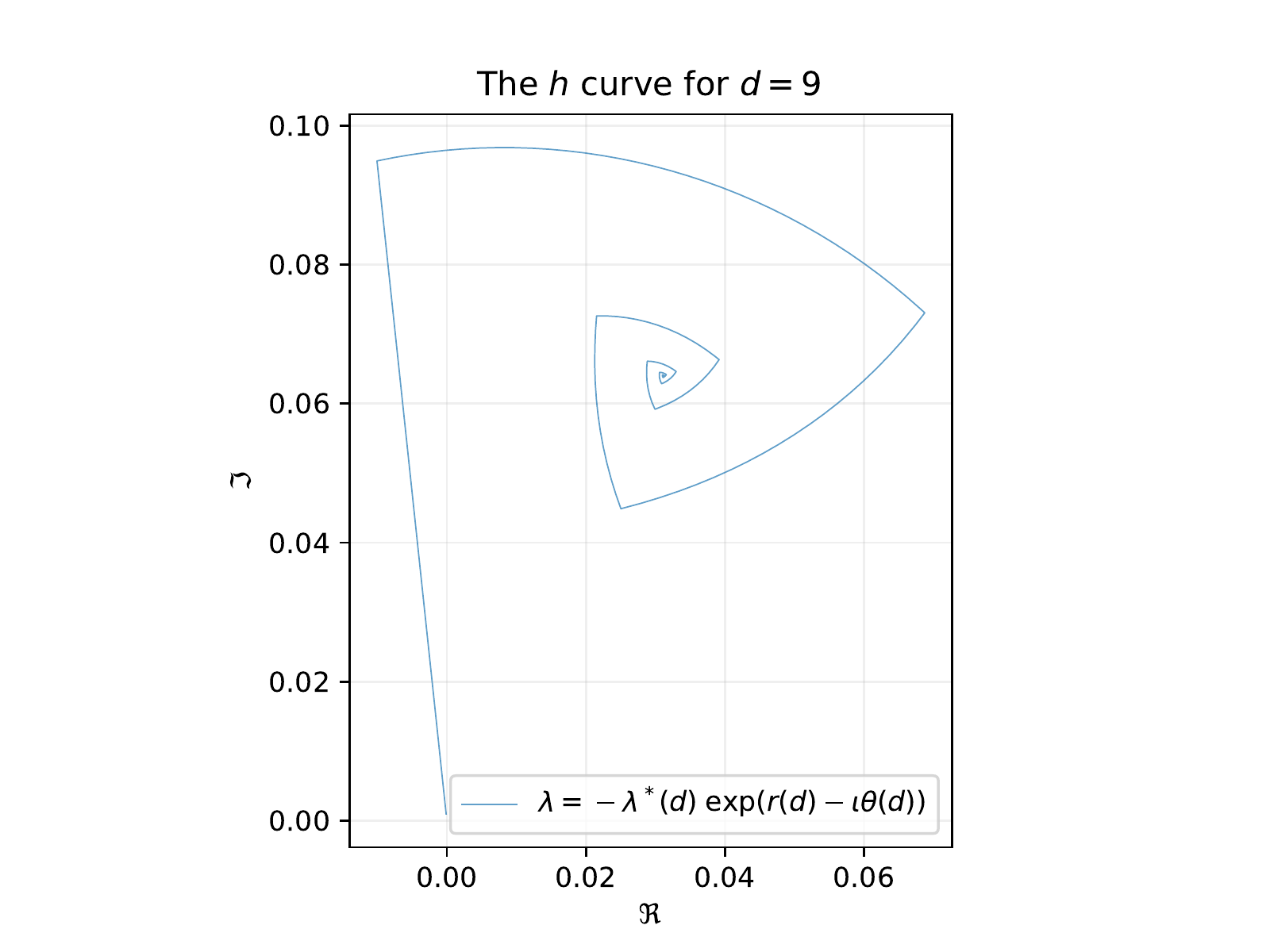}
  \caption{The $h$ curve from \cref{thm:Simons-result}. We set $d=9$, and
    $\lambda^*(d)$ is the Shearer threshold described in
    \Cref{thm:shearer}. $\theta(d)$ is chosen to be
    $\cos^{-1}\frac{1}{d + 0.5}$, while
    $r(d) \defeq \min\inb{d\log (1 + \frac{1}{d}),
      \frac{2d(d+1)\sin^2(\theta/2)}{d^2 + 4(d+1)\sin^2(\theta/2)}}$.  The
    choice of these parameters is based on their use in the later
    \cref{thm:critical-vicinity}, which applies \cref{thm:Simons-result}.}
  \label{fig:h-curve}
\end{figure}

\begin{theorem}
\label{thm:Simons-result}
For $\lambda \in \CC$, $\Im(\lambda) > 0$, define a curve
\begin{itemize}
\item $h(t) = t \lambda$ for $t \in [0,1]$,
\item $h(t) = \frac{\lambda}{(1+h(t-1))^d}$ for $t > 1$.
\end{itemize}
If $\Im(h(t)) \geq 0$ for all $t \geq 0$, then $Z_G(\lambda) \neq 0$ for all
$G \in \cG_{d+1}$.  (See also \Cref{fig:h-curve} for an example of the curve
$h$ in the statement of the theorem.  )
\end{theorem}

\begin{proof}
  Note that the curve $h(t)$ is continuous since it is continuous at $t\leq 1$:
  this is because at $t=1$ we have $\lim_{t\to 1-}h(t)=\lim_{t\to 1+}h(t)=\lambda$, and by the recursion $h$ is continuous at $t$ if it is continuous at $t-1$.
  By continuity and  the assumption
  $\Im(h(t)) \geq 0$ for all $t \geq 0$ this implies that the continuous
  functions $\arg (h(t))$ and $\arg(1 + h(t))$ are also non-negative for
  $t \geq 0$.  The identity $\arg(h(t)) = \arg(\lambda) - d \arg(1+h(t-1))$ for
  $t \geq 1$ (valid whenever the right hand side lies in $(-\pi, \pi]$) then
  implies that $\arg(1+h(t-1))$ cannot exceed $\frac{1}{d} \arg(\lambda)$,
  because if $t$ is the infimum of points for which
  $\arg(1+h(t-1)) > \frac{1}{d} \arg(\lambda)$ then the identity gives a
  contradiction to $\arg(h(t + \epsilon)) \geq 0$ for some small enough positive
  $\epsilon$.  Since $\arg(1+h(t-1)) \geq 0$, the identity then also implies
  that
  \begin{equation}
    \label{eq:8}
    0 \leq \arg(h(t)) = \arg(\lambda) - d \arg(1+h(t-1)) \leq \arg(\lambda).
 \end{equation}
 Hence the argument of any point of the curve is contained in
  $[0,\arg(\lambda)]$.

  In order to define a trapping region, we start by defining the following
  quantity:
\begin{equation}
    \label{eq:9}
    \tau^\star \defeq \sup
    \{
    t': \arg(1+h(t)) \text{ is non-decreasing for all } t
    \in [0,t']
    \}.
  \end{equation}
  Note that $\tau^\star \geq 1$.  We also allow $\tau^\star=\infty$, although this cannot
  really happen.  We will now show that the region $S$ in the upper half plane
  bounded by the line segments $[0, h(1)]$ and $[0, h(\tau^\star + 1)]$ and the curve
  $\inb{h(t + 1) | 0 \leq t \leq \tau^\star}$ is a trapping region in the sense of
  \Cref{thm:z-criterion}.  Note that by definition $0 \in S$, while
  $-1 \not\in S$ (since, from the observations above, $\arg z \in [0, \arg
  \lambda]$ for all $z \in S$).  It remains to show that (1) $S$ is convex, and
  (2) $f(z) = \lambda/(1+z)^d \in S$ for all $z$ in $S$.

  \newcommand{\Dr}[1][]{\ensuremath{D^+\notblank{#1}{\inp{#1}}{}}}
  \newcommand{\Dl}[1][]{\ensuremath{D^-\notblank{#1}{\inp{#1}}{}}} We start by
  proving that $S$ is convex. Since the curve $h$ lies in the upper half plane
  (see \cref{eq:8}), this will follow if we establish the following two facts:
  \begin{enumerate}
  \item $\arg (h(t))$ is non-increasing for $t \in [1, \tau^\star + 1]$.\label{item:5}
  \item The curve $\inb{h(t) :  t \in [1, \tau^\star + 1]}$ is ``turning to the
      right''.  More formally, for any $t \in [1, \tau^\star + 1)$, there is a small
      enough neighborhood $N_t$ of $t$ such that for $t_1 \leq t_2$ in $N_t$,
      $\arg (\Dl[h](t_1) / \Dr[h](t_2)) \geq 0$.  Here $\Dr$ and $\Dl$ denote the
      right and left one-sided derivatives.\label{item:6}
  \end{enumerate}

  We first prove \cref{item:5}.  This follows since in the interval
  $t \in [1,\tau^\star+1]$, we have $\arg(h(t)) = \arg(\lambda) - d \arg(1+h(t-1))$,
  which is non-increasing by the definition of $\tau^\star$.

  We now consider \cref{item:6}.  Note that $h(t)$ is continuously
  differentiable in the neighborhood of any $t$ which is not an integer.
  Further, for such a $t$, we have
  $h'(t) = \frac{-\lambda d}{(1 + h(t-1))^{d+1}}h'(t-1)$.  We now prove the
  claim for such $t$ (i.e., non-integral $t$) using an induction on $\ceil{t}$.
  In the base case, when $\ceil{t} = 1$, we have $h'(t) = \lambda$, so
  $\arg (\Dl[h](t_1)/\Dr[h](t_2)) = \arg(1) = 0$ for $t_1, t_2$ in any small enough
  neighborhood of $t$.  In the inductive case, we have, for $t_1 \leq t_2$ in a
  small enough neighborhood of $t$,
  \begin{equation}
    \arg \frac{\Dl[h](t_1)}{\Dr[h](t_2)} = (d+1) \arg \frac{1 + h(t_2 - 1)}{1 +
    h(t_1 - 1)} + \arg \frac{\Dl[h](t_1-1)}{\Dr[h](t_2-1)}.\label{eq:15}
  \end{equation}
  The claim now follows since the first term is non-negative due to the
  definition of $\tau^\star$, while the second is non-negative by the inductive
  hypothesis.

  We now consider the case of integral $t$.  Here, we find via a direct
  induction that
  \begin{equation}
    \arg \frac{\Dl[h](t)}{\Dr[h](t)} = \arg \frac{\Dl[h](1)}{\Dr[h](1)} = \arg
    \frac{\lambda}{(-d\lambda^2)} = \pi -
    \arg \lambda \geq 0.\label{eq:30}
  \end{equation}
  The proof for \cref{item:6} now follows from the already proved case of
  non-integral $t$ and the fact that the derivative $h'$ is a well-defined
  continuous function except at integral $t$.  As noted earlier, this proves
  that $S$ is convex.  In fact, from the definition of $\tau^\star$, we also obtain
  that $0 = \arg (1 + h(0)) \leq \arg (1 + z) \leq \arg (1 + h(\tau^\star))$ for all
  $z \in S$.  Since $\arg(1 + h(t))$ is non-decreasing for $t \in [0, \tau^\star]$ and
  $S$ is convex, it follows that if a line is drawn from $z \in S$ in the
  direction of $-1$, it will intersect the boundary of $S$ at some point $h(t)$
  for $0 \leq t \leq \tau^\star$.

  We can now prove that $S$ satisfies the remaining requirement for being a
  trapping region, which is, that it is closed under application of $f$.  Consider
  any point $z \in S$.  As noted above, if a line is drawn from $z$ towards
  $-1$, then it must intersect the boundary of $S$ on a point $\tilde{z}$ of the
  form $h(t)$ for $t \in [0,\tau^\star]$.  Hence, there is a point
  $\tilde{z} = h(t), t \in [0,\tau^\star]$ such that $\arg(1+h(t)) = \arg(1+z)$ and
  $|1+h(t)| \leq |1+z|$. By construction,
  \begin{equation}
    \label{eq:38}
    f(h(t)) = \frac{\lambda}{(1+h(t))^d} = h(t+1)
  \end{equation}
  which is still in $S$ (since $t \in [0,\tau^\star]$). Finally,
  $f(z) = \frac{\lambda}{(1+z)^d}$ has the same argument as $h(t+1)$, and
  possibly smaller modulus, hence $f(z) \in S$ by convexity (since $0 \in S$).

  Thus, $S$ as defined above is a trapping region, and this concludes the proof.
\end{proof}

\begin{remark}\label{rem:numerical}
  We note that the proof of \cref{thm:Simons-result} also indicates a numerical
  approach to check this criterion, see
  Figure~\ref{fig:Full_power_iterated_spiral}. We do not have to track the curve
  for $t \rightarrow \infty$. It is sufficient to compute $h(t)$ for
  $t \in [0,\tau^\star+1]$ as defined above. If
  $\arg(1+h(\tau^\star)) \leq \frac{1}{d} \arg(\lambda)$ and $\arg(1+h(t))$ is
  non-increasing for $t \in [\tau^\star,\tau^\star+1]$, the argument above
  implies that $S$ is a trapping region and the entire curve is contained in the
  upper half-plane.
\end{remark}
\begin{remark}\label{rem:alternative}
  We also remark that it is possible to prove \cref{thm:Simons-result} from
  \cref{thm:init-curve} by considering the curve from the proof of
  \cref{thm:Simons-result} on the interval $[0,\tau^\star]$. The curve in
  \cref{thm:init-curve} represents the portion of the curve in
  \cref{thm:Simons-result} ``visible from $-1$''; i.e., the points $h(t)$ that
  are not $-1$-covered by any other point $h(t')$.
\end{remark}

\begin{figure}
    \centering
    \includegraphics[width=0.8\textwidth]{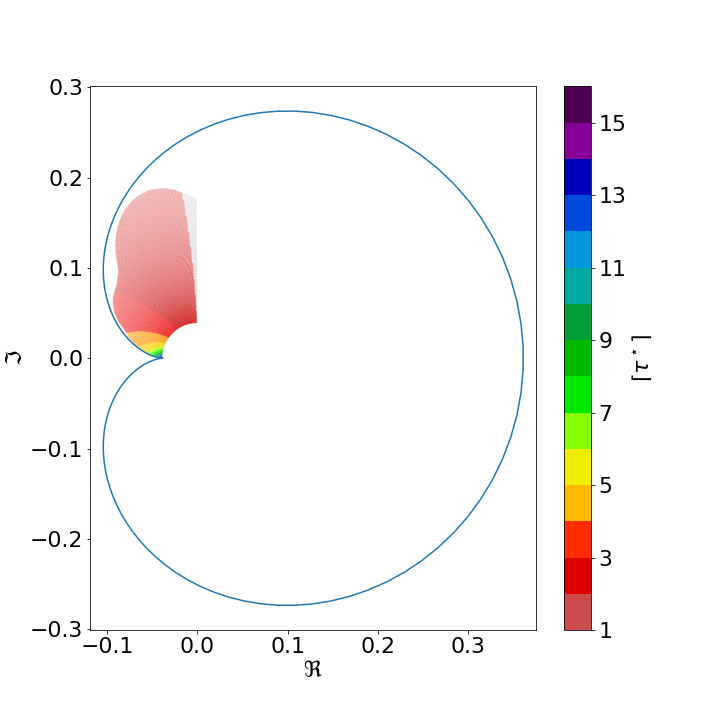}
    \caption{A numerical exploration of points that satisfy the condition of
      Theorem~\ref{thm:Simons-result} for $d=9$. Colors represent the value of
      $\ceil{\tau^\star}$.}
    \label{fig:Full_power_iterated_spiral}
\end{figure}

\section{Derivation of the Sokal conjecture} \label{sec:sokal-conjecture}

Here we provide a short proof using \cref{thm:init-curve} that there are no complex roots close to the positive real axis, up to the critical point $\lambda^* = \frac{d^d}{(d-1)^{d+1}}$. This was first proved by Peters and Regts \cite{peters_conjecture_2019}.

\begin{figure}[t]
  \centering
  \includegraphics[scale=0.8]{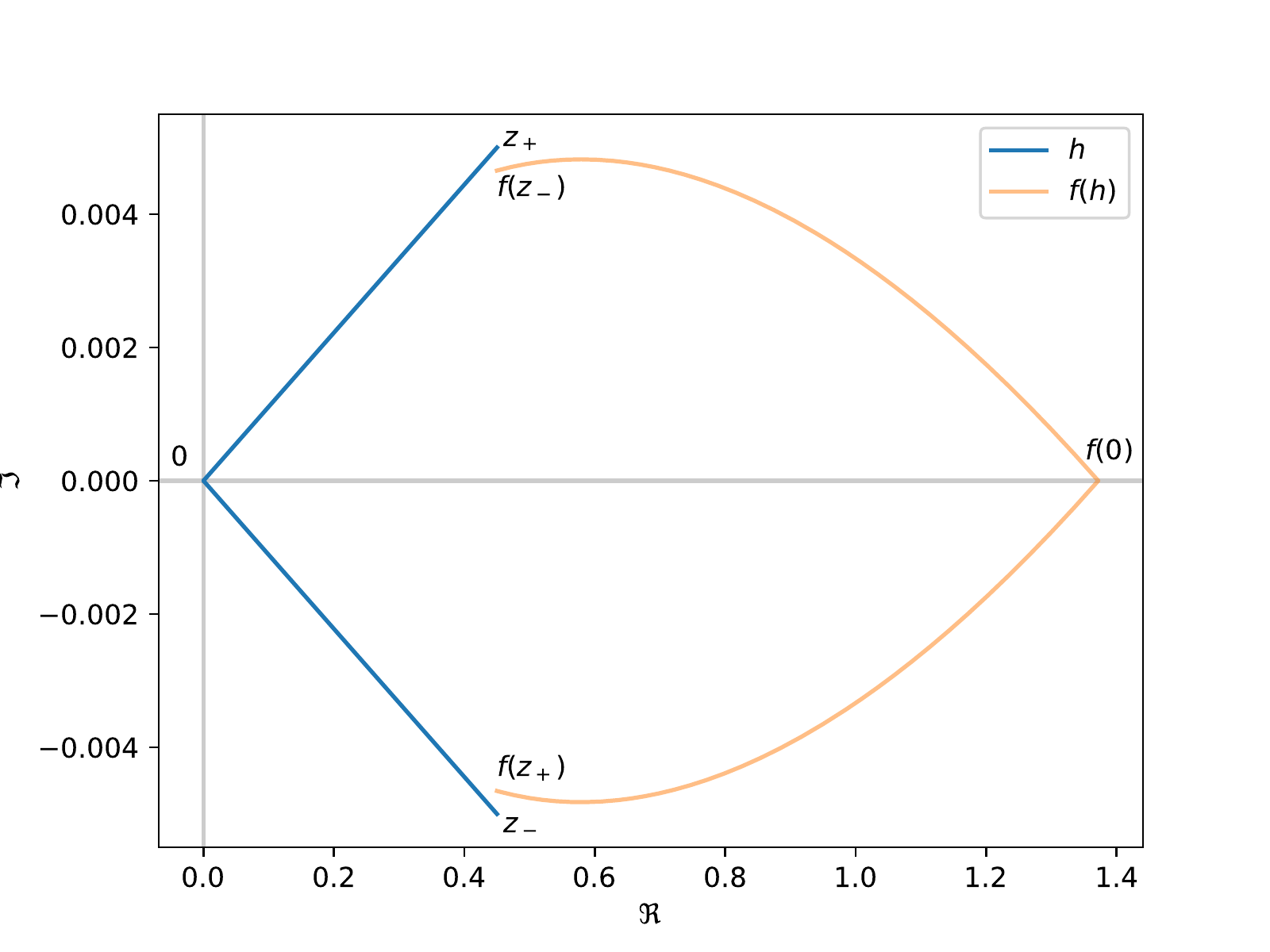}
  \caption{A sketch of the $h$ curve in the proof of \cref{thm:sokal-conj}.  In
    the notation of the theorem, the sketch corresponds to $d \defeq 3$,
    $\epsilon \defeq 0.1$ and $\delta \defeq 0.01$.  $\theta$ has been set to
    $0$ for simplicity. As before, $f$ is the map $z \mapsto \lambda/(1+z)^d$.
    The aspect ratio in the figure has been chosen to be different from $1$ to
    accentuate features close to the real line.}
  \label{fig:sokal-h}
\end{figure}

\begin{theorem}
  For every fixed $d \geq 2$, and every $\epsilon \in (0, 1)$, there is an
  $\epsilon' > 0$ such that $Z_G(\lambda) \neq 0$ for $G$ of maximum degree at
  most $d + 1$ when
  $\lambda = (1-\epsilon)\frac{(d-\epsilon)^d}{(d-1)^{d+1}}\exp(\iota \theta)$
  with $\abs{\theta} \leq \epsilon'$.\label{thm:sokal-conj}
\end{theorem}

\begin{proof}
  Note that $z_0 \defeq \frac{1-\epsilon}{d-1}$ is a fixed point of the map
  $z \mapsto \frac{\lambda_0}{(1+z)^d}$ where
  $\lambda_0 = (1-\epsilon)\frac{(d-\epsilon)^d}{(d-1)^{d+1}}$.  Note also that
  $\lambda_0 \rightarrow \lambda^*$ as $\epsilon \rightarrow 0$. We consider two
  complex conjugate points $z_\pm = \frac{1-\epsilon \pm \iota \delta}{d-1}$ for
  some $\delta>0$ to be fixed later.  We have
  $\lambda = \lambda_0\exp(\iota\theta)$ where $\theta > 0$ is small enough (as
  a function of $\epsilon$ and $\delta$) to be fixed later.  We will now use
  \Cref{thm:init-curve} for the curve defined by
  \begin{equation}
    h(t)=
    \begin{cases}
      -t z_- & \text{if } t \in [-1, 0)\\
      t z_+ & \text{if } t \in [0, 1]
    \end{cases}.
  \end{equation}
  The first three conditions required of the curve $h$ in \Cref{thm:init-curve}
  are satisfied by construction (see \cref{fig:sokal-h} for an example sketch).
  We now proceed to verify the fourth condition.  To start with, a direct
  computation reveals that
  \begin{align}
    f(z_-) = \frac{\lambda}{(1+z_-)^d}
    &= (1-\epsilon)
      \frac{(d-\epsilon)^d}{(d-1)^{d+1}}
      \left(1 + \frac{1-\epsilon-\iota \delta}{d-1}\right)^{-d}\exp(\iota \theta)\\
    &= \frac{1-\epsilon}{d-1}
      \left(1
      - \frac{\iota \delta}{d-\epsilon}
      \right)^{-d} \exp(\iota \theta).\label{eq:13}
  \end{align}
  We have
  $\arg(1 - \frac{\delta i}{d-\epsilon}) = -\tan^{-1}
  (\frac{\delta}{d-\epsilon}) = -\frac{\delta}{d-\epsilon} + O(\delta^2)$. Therefore,
  \begin{equation}
    \arg(f(z_-)) = -d \arg \inp{1 - \frac{\iota \delta}{d-\epsilon}} + \theta
    = \frac{d \delta}{d-\epsilon} + O(\delta^2) + \theta.\label{eq:17}
  \end{equation}
  
  In comparison,
  \begin{equation}
    \arg z_+  = \tan^{-1} \left( \frac{\delta}{1-\epsilon} \right) =
    \frac{\delta}{1-\epsilon} + O(\delta^2).
     \label{eq:19}
  \end{equation}
  Thus, for all $\delta>0$ small enough (depending on $\epsilon$ and $d$) and
  all $\theta \geq 0$ small enough (depending on $d \geq 2$, $\delta$ and
  $\epsilon$), we have
  \begin{equation}
    \arg z_+ 
    > \arg f(z_-) > 0.\label{eq:16}
  \end{equation}
  Observe also that $|f(z_-)| < | \frac{1-\epsilon}{d-1} | < |z_+|$. This
  implies that $f(z_-)$ is $-1$-covered by some point on the line segment from
  $0$ to $z_+$, and in particular, $0 \leq \arg(1 + f(z_-)) \leq \arg(1 + z_+)$.
  An essentially symmetric argument shows that $f(z_+)$ is also $-1$-covered by
  some point on the line segment between $0$ and $z_-$, and in particular,
  $0 \geq \arg(1 + f(z_+)) \geq \arg(1 + z_-)$.  The inequality analogous to
  \cref{eq:16} for $f(z_+)$ is (again, provided that $\theta$ has small enough
  magnitude)
  \begin{equation}
    \arg z_- 
    < \arg f(z_+) < 0.\label{eq:20}
  \end{equation}
  For later use, we also record the following computation. For small positive
  $\theta$ and $\delta$, we have, by a direct computation,
  \begin{align}
    \label{eq:25}
    \arg \frac{f(z_+)}{1 + f(z_+)}
    &= -\frac{d(d-1)}{(d-\epsilon)^2}\delta +
      O(\delta^2) + O(\theta), \text{ and }\\
    \label{eq:26}
    \arg \frac{z_+}{1 + z_+}
    &= \frac{d-1}{(d-\epsilon)(1-\epsilon)}\delta +
      O(\delta^2).
  \end{align}
  In particular, since $d \geq 2$, we have
  \begin{equation}
    \label{eq:27}
    \arg \frac{f(z_+)}{1 + f(z_+)} +     \arg \frac{z_+}{1 + z_+} \geq 0
  \end{equation}
  for all small enough $\delta > 0$ and $\theta > 0$.  At this point, we specify
  our choice of $\theta$ and $\delta$: we choose $\delta < 1$ and
  $\theta < \pi/10$ positive and small enough that (i) \cref{eq:16}, its
  analogue \cref{eq:20} for $f(z_+)$, and \cref{eq:27} are all valid, and (ii)
  $\abs{\arg{z_+}} = \abs{\arg{z_-}} \leq \pi/10$.  In the following, we use
  these conditions imposed on $\delta$ and $\theta$ without comment.

  We now claim that the curve $\{ f(t z_+): 0 \leq t \leq 1 \}$ is also
  $-1$-covered by the curve $h$ defined above.  To prove this, we define
  $\gamma(t) \defeq 1 + f(tz_+)$, and note that we have (for $s \in (0, 1)$)
  $\diff{}{s} \arg\gamma(s)\vert_{s = t} = \Im \frac{\gamma'(t)}{\gamma(t)}$.
 
  We begin by noting that as $t$ increases from $0$ to $1$, $\arg f(tz_+)$
  decreases from $\theta = \arg f(0)$ to $\arg f(z_+) > \arg z_-$ (since
  $\arg\inp{1 + tz_+}$ increases as $t$ increases), while $\abs{f(tz_+)}$
  decreases from $\lambda_0$ to $\abs{f(z_+)} < \abs{z_-}$ (again, since
  $\abs{1 + tz_+}$ increases as $t$ increases). We now compute
  \begin{equation}
    \diff{}{t} \arg\gamma(t) = \Im \frac{\gamma'(t)}{\gamma(t)}
    = -d\Im \inp{\frac{z_+}{1 + tz_+}\cdot\frac{f(tz_+)}{1 + f(tz_+)}}.\label{eq:18}
  \end{equation}
  We will now show that $\diff{}{t} \arg \gamma(t) \leq 0$ for all
  $t \in (0, 1)$.  For any particular $t$, if $\arg f(tz_+) \geq 0$, then
  \cref{eq:18} immediately implies that $\diff{}{t} \arg \gamma(t) < 0$ (since
  $\arg z_+ > 0$).  For any other $t \in (0, 1)$, we must have
  $\arg f(z_+) \leq \arg f(t z_+) \leq 0$.  We
  then get
  \begin{align}
    \arg{\frac{f(tz_+)}{1 + f(tz_+)}}
    &= \tan^{-1} \frac{
      \sin \arg f(tz_+)
      }{
      \abs{f(tz_+)} + \cos \arg f(tz_+)
      }\\
    &\ge \tan^{-1} \frac{
      \sin \arg f(z_+)
      }{
      \abs{f(tz_+)} + \cos \arg f(z_+)
      }\quad \text{ since } \arg f(t{z_+}) \geq \arg f(z_+),\\
    &\ge \tan^{-1} \frac{
      \sin \arg f(z_+)
      }{
      \abs{f(z_+)} + \cos \arg f(z_+)
      }\quad \text{ since } \abs{f(t{z_+})} \geq \abs{f(z_+)} \text{ and } \arg
      f(z_+) \leq 0,\\
    & = \arg \frac{f(z_+)}{1 + f(z_+)} \geq  -\arg\frac{z_+}{1 + z_+}.
  \end{align}
  Here, the last inequality comes from \cref{eq:27}.  Combining this with the
  observation that $\arg \frac{z_+}{1 + t z_+}$ is strictly decreasing in $t$
  for $t \in (0, 1)$, and substituting in \cref{eq:18}, we get the required
  claim that $\diff{}{t} \arg \gamma(t) \leq 0$ for all $t$ in $(0, 1)$.  Thus,
  $\arg(1 + f(tz_+))$ \emph{decreases} as $t$ increases from $0$ to $1$.  An
  essentially symmetrical argument shows that $\arg(1 + f(tz_-))$
  \emph{increases} as $t$ increases from $0$ to $1$.  Since we already
  established that $f(z_+)$ and $f(z_-)$ are $-1$-covered by $h$, and also that
  $\arg z_-  < \arg f(h(t))  < \arg z_+ $ for all $t \in [-1, 1]$, this
  establishes that the whole curve $\inb{f(h(t)) \; \vert \; t \in [-1, 1]}$ is
  $-1$-covered by $h$.

  We thus see that the fourth condition of \Cref{thm:init-curve} is also
  satisfied for the curve $h$.  We conclude therefore that $Z_G(\lambda) \neq 0$
  for all graphs $G$ of maximum degree at most $d+1$.

\end{proof}

\section{A new zero-free region in the vicinity of the critical point}
\label{sec:zero-free-region-negative-real-line}
 In this section, we use \Cref{thm:Simons-result} to establish a new zero-free
region for the independence polynomial in the vicinity of the negative real
line.  The result in this section applies more generally to points in the left
half-plane away from the imaginary axis; we consider points close to the
imaginary axis in \Cref{sec:zero-free-region-imaginary-axis}.

We recall that $\lambda^* = \lambda^*(d) \defeq \frac{d^d}{(d+1)^{d+1}}$ is the
Shearer threshold.  Consider the boundary $\partial U_d$ of the
``cardioid-shaped'' region $U_d$ (\cref{eq:12}) of Peters and
Regts~\cite{peters_conjecture_2019}.  Near the negative real line, one can
calculate that the curve $\partial U_d$ follows a power law of the following
form.  Let $R_U(\theta)$ denote the polar equation of $\partial U_d$, and let
$(X_U(\theta), Y_U(\theta))$ denote the corresponding Cartesian coordinates
$(R_U(\theta)\cos\theta, R_U(\theta)\sin\theta)$.  In the vicinity of the point
$-\lambda^*(d)$ on $\partial U_d$, a somewhat tedious but straightforward
calculation shows that for small $\phi$,
\begin{equation}
  X_U(\pi + \phi) = -\lambda^*(d) - c_d \cdot \abs{\phi}^{2/3} +
  o(\abs{\phi)}^{2/3}),\label{eq:39}
\end{equation}
where $c_d$ is a positive constant depending only on $d$.  While we cannot prove
that the true root-free region matches this exact power law, we have the
following result which gives a weaker power law (see \cref{fig:1} for a
pictorial description).

\begin{theorem}
  \label{thm:critical-vicinity}
  Fix an integer $d \geq 2$.  If $\lambda = -\lambda^*\exp(r - \im \theta)$,
  where $\theta \in (0, \cos^{-1}\frac{1}{d + 0.5}]$ and
  $0 \leq r \leq \min\inb{d\log (1 + \frac{1}{d}),
    \frac{2d(d+1)\sin^2(\theta/2)}{d^2 + 4(d+1)\sin^2(\theta/2)}}$, then
  $Z_G(\lambda) \neq 0$ for any graph $G$ of degree at most $d +
  1$.
\end{theorem}

\begin{figure}[t]
  \centering
  \includegraphics[width=0.6\textwidth]{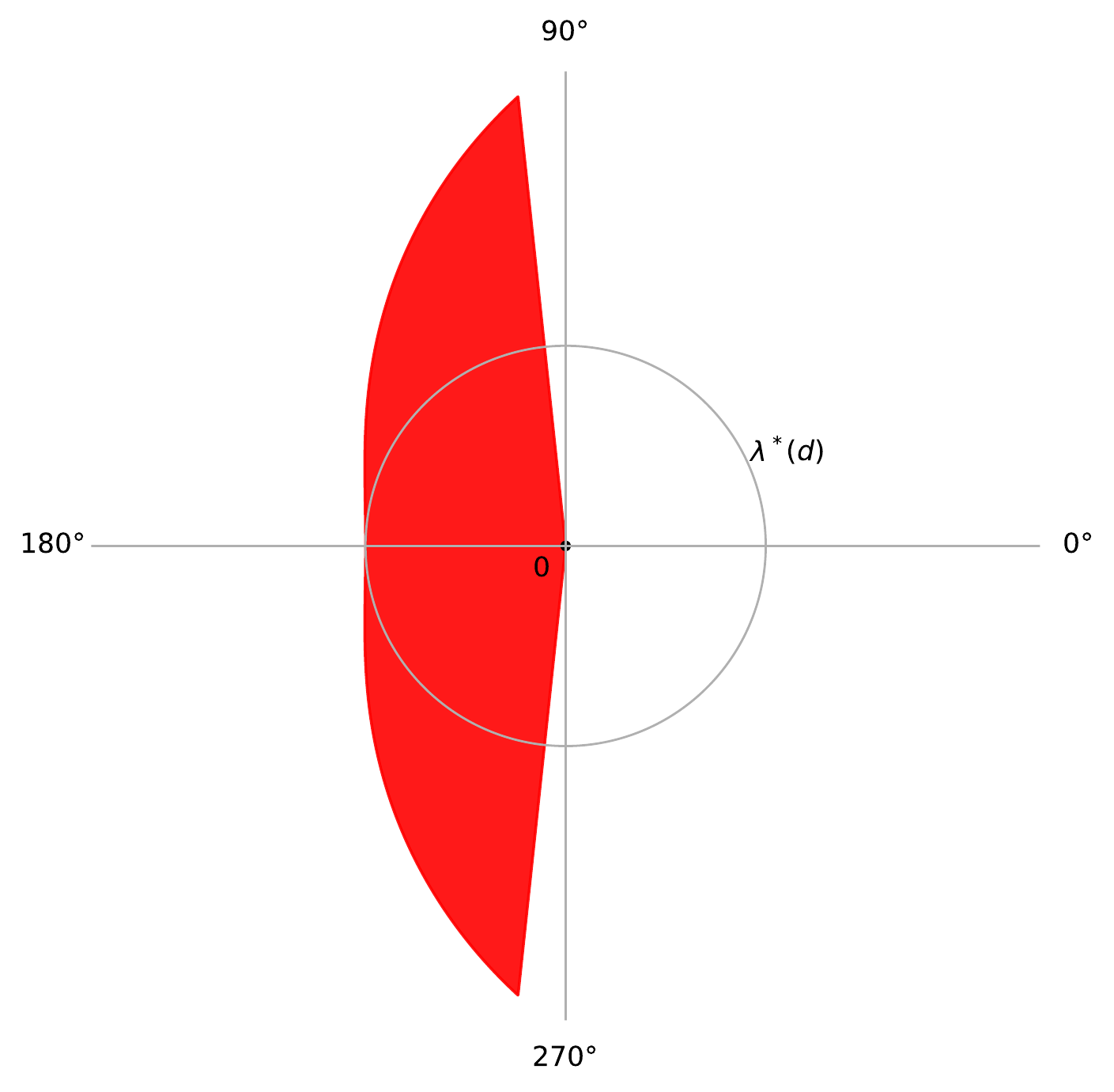}
  \caption{The red region is the zero-free region in
    \cref{thm:critical-vicinity}, plotted here for $d = 9$ (i.e., for graphs of
    degree at most $10$).  The black circle around the origin has radius
    $\lambda^*(d)$, and the markings are according to polar coordinates.}
  \label{fig:1}
\end{figure}

Before proving the theorem, we briefly describe the power law (analogous to the
one stated above for $\partial U_d$) that the region described in the theorem
follows.  Again, we denote by $\tilde{R}(\theta)$ the polar equation of the
boundary of the region described by the theorem, and let
$(\tilde{X}(\theta), \tilde{Y}(\theta))$ denote the corresponding Cartesian
coordinates $(\tilde{R}(\theta)\cos\theta, \tilde{R}(\theta)\sin\theta)$.  In
the vicinity of the point $-\lambda^*(d) = R_U(\pi) = \tilde{R}(\pi)$, a similar
computation as above then shows that for small $\phi$,
\begin{equation}
   \tilde{X}(\pi + \phi) = -\lambda^*(d) - \tilde{c}_d\cdot \abs{\phi}^{2} +
  o(\abs{\phi}^{2}),
\end{equation}
where $\tilde{c}_d = \frac{\lambda_*(d)}{2d}$ is a positive constant depending
only on $d$.

\begin{proof}[Proof of \Cref{thm:critical-vicinity}]
  We use \Cref{thm:Simons-result}.  In particular, we will show that the curve
  $h(t)$ defined there lies in the upper half plane $\inb{z | \Im z \geq 0}$ for
  every $t \geq 0$. (See \Cref{fig:h-curve} for an example of this curve for a
  particular setting of the parameters $d, r$ and $\theta$.) In fact, we will
  prove by an induction on $\ceil{t}$ that for all $t \geq 0$,
  \begin{enumerate}
  \item \label{item:1}
    $\abs{h(t)} \leq \tau \defeq \frac{1}{d + 1 - \delta} \leq \frac{1}{d + 0.5}
    \leq \frac{1}{2}$, and
  \item \label{item:2} $\arg{h(t)} \in [0, \pi - \theta]$,
  \end{enumerate}
  where $\delta = \delta(d, r, \theta) < 1/2$ is a fixed non-negative constant.  We
  first verify these for the base case $\ceil{t} = 1$.  In this case, we have
  $\abs{h(t)}$ =
  $t\abs{\lambda} \leq \exp(r)\frac{d^d}{(d+1)^{d+1}} \leq \frac{1}{d+1} \leq
  \tau$ since $r \leq d\log(1+1/d)$.  Further,
  $\arg{h(t)} = \arg{\lambda} = \pi - \theta$.

  We now proceed with the induction.  For ease of notation, we denote
  $\abs{h(t-1)}$ as $\rho$ and $\arg h(t-1)$ as $\pi - \alpha$.  From the
  induction hypothesis, we have $\rho \leq \tau$ and
  $\alpha \in [\theta, \pi]$.  This gives
  \begin{align}
    \abs{1 + h(t-1)}
    &= \sqrt{1 + \rho^2  - 2\rho\cos\alpha}
      \geq \sqrt{1 + \rho^2  - 2\rho\cos\theta} \nonumber\\
    &\geq\sqrt{1 + \tau^2  - 2\tau\cos \theta} = (1-\tau)\sqrt{1 + 2\tau\cdot\frac{2\sin^2(\theta/2)}{(1-\tau)^2}}\nonumber\\
    &\geq \frac{d-\delta}{d+1-\delta}\cdot\sqrt{1 + \frac{4(d+1)\sin^2(\theta/2)}{d^2}}.\label{eq:3}
  \end{align}
  Here, for the second inequality we use the fact that the quantity inside the
  square-root is decreasing in $\rho$ since
  $\rho \leq \tau \leq \frac{1}{d + 0.5} \leq \cos \theta$, since
  $\theta \in [0, \cos^{-1}(1/(d+0.5))]$.  Similarly, the last inequality uses
  $\tau \geq \frac{1}{d+1}$. Now, note that since $\arg h(t-1) \in [0, \pi]$,
  $\abs{h(t-1)} \leq \tau \leq 1$ and $\abs{1 + h(t-1)} > 0$, we have
  \begin{equation}
    \arg (1 + h(t-1)) \geq 0,\label{eq:2}
  \end{equation}
  and also
  \begin{align}
    \arg(1+h(t-1)) &\leq \frac{\Im h(t-1)}{1 + \Re h(t-1)} = \frac{\rho\sin
                     \alpha}{1 - \rho \cos\alpha} \leq \frac{\tau\sin
                     \alpha}{1 - \tau \cos\alpha}\nonumber\\
    &=\frac{\sin \alpha}{d + 1 - \delta - \cos \alpha}.
  \end{align}
  From this, using the fact that
  $d \sin \alpha + \cos \alpha \leq \sqrt{d^2 + 1}$ for $\alpha \in [0, \pi]$,
  we deduce that
  \begin{equation}
    \arg(1+h(t-1)) \leq 1/d,\label{eq:1}
  \end{equation}
  provided that $\delta \leq 1/2$.  Now, we have
  \begin{equation}
    \arg h(t) = \arg \lambda - d\arg (1 + h(t-1)) = \pi - \theta - d\arg(1+h(t-1)),
  \end{equation}
  so that \cref{eq:1,eq:2} imply \cref{item:2} of the induction hypothesis
  (since $\theta \leq \pi/2$ so that $\pi - \theta - 1 \geq 0$).  For
  \cref{item:1}, we use \cref{eq:3} to calculate
  \begin{align}
    \log \abs{h(t)} + \log (d + 1 - \delta)
    & = r + d\log d - (d+1)\log(d+1) - d \log \abs{1 + h(t-1)} + \log (d + 1 - \delta)
    \nonumber\\
    &\leq r + (d+1)\log\inp{1 -
      \frac{\delta}{d+1}} - d\log \inp{1 - \frac{\delta}{d}}\\
    &\qquad\qquad - \frac{d}{2}\log \inp{1 + \frac{4 (d+1) \sin ^2 (\theta/2)}{d^2}}\\
    &< r - \delta + \frac{d\delta}{d - \delta} -
      \frac{2d(d+1)\sin^2(\theta/2)}{d^2 + 4(d+1)\sin^2(\theta/2)}\\
    &= r - \frac{2d(d+1)\sin^2(\theta/2)}{d^2
      + 4(d+1)\sin^2(\theta/2)} + \frac{\delta^2}{d-\delta} \leq 0,
  \end{align}
  provided $r \leq \frac{2d(d+1)\sin^2(\theta/2)}{d^2 + 4(d+1)\sin^2(\theta/2)}$
  and $\delta$ is chosen to be a small enough non-negative constant depending only
  upon $r, d$ and $\theta$ (note that the second inequality above is strict when
  $\delta$ is positive).
\end{proof}

\section{A zero-free region close to the imaginary axis}
\label{sec:zero-free-region-imaginary-axis}
The analysis in the previous section was devoted to understanding the behavior
of the zero-free region close to the negative real line.  We now turn to
understanding the behavior of the zero region close to the imaginary axis.  The
theorem below, while it covers all arguments in the third argument, is most
interesting when the argument of the activity $\lambda$ is closer to $\pi/2$
than to $\pi$.

  \begin{figure}[t]
    \centering
    \includegraphics[scale=0.8]{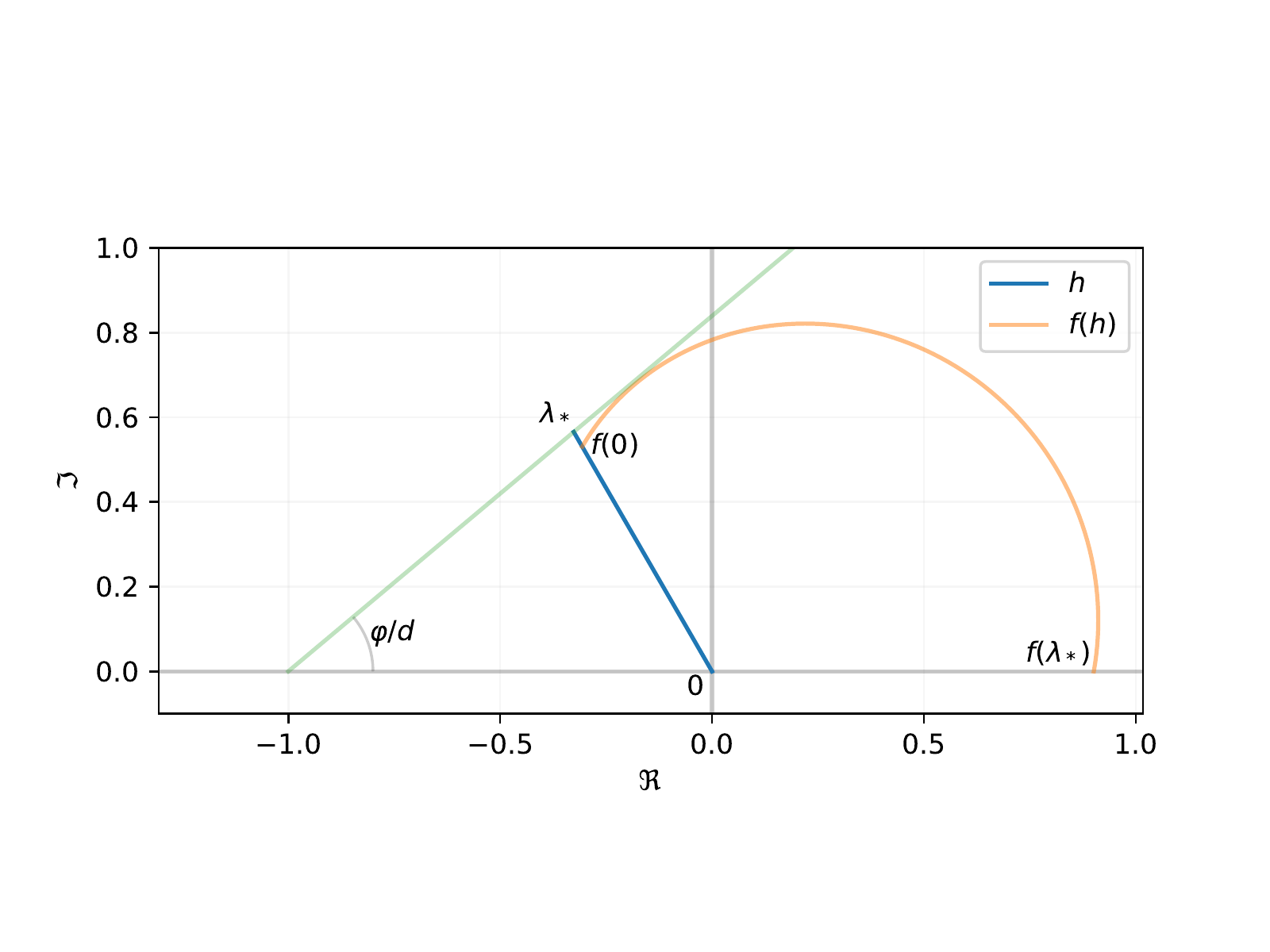}
    \caption{A sketch of the $h$ curve in the proof of \cref{thm:lhp}.  In the
      notation of the theorem, the sketch corresponds to $d \defeq 3$ and
      $\varphi \defeq 2\pi/3$, and $r$ has been chosen to be $0.99$ times the
      value on the right hand side of \cref{eq:b}. As before, $f$ is the map
      $z \mapsto \lambda/(1+z)^d$.}
    \label{fig:close-to-imaginary-h}
  \end{figure}

\begin{theorem}
Let $d \geq 2$.  Suppose that $\lambda = r e^{\iota \varphi}$ where
  $\varphi \in [\frac{\pi}{2}, \pi)$ and
  \begin{align}\label{eq:b}
    r < \frac{\sin(\varphi/d)\sin^d(\varphi)}{\sin((d-1)\varphi/d-d\psi^*)\sin^d(\varphi-\psi^*)},
 \end{align}
 where $\psi^*=  \max\left(\frac{1}{d+1}\left((2-1/d)\varphi-\pi\right),0\right)$.
 Then $Z_G(\lambda) \neq 0$ for any graph of degree at most   $d+1$.\label{thm:lhp}
\end{theorem}
\begin{proof}
  For a given $\varphi$, let us choose $r_*,\lambda_*$ so that
  $\arg(1+\lambda_*)=\varphi/d$ and $\arg(\lambda_*)=\arg(\lambda)=\varphi$.  Note that for
  $t\in[0,1]$, $\varphi/d\ge \arg(1+t\lambda_*)\ge 0$. As before denote
  $f(z)=\frac{\lambda}{(1+z)^d}$.
  
  We claim that the function $h(t)=t\lambda_*$ satisfies the conditions of
  \cref{thm:init-curve}. The first three conditions of the theorem are satisfied
  trivially, thus we only have to show that the points $f(t\lambda_*)$,
  $t \in [0, 1]$ are $-1$-covered by the segment
  \begin{equation}
    \{t\lambda_* ~:~t\in[0,1] \}.\label{eq:28}
  \end{equation}
  Further, as $\arg(f(t\lambda_*))=\varphi-d\arg(1+t\lambda_*)$ decreases
  monotonically from $\varphi$ to $0$ as $t$ goes from $0$ to $1$, it would be
  sufficient to prove that $\arg(1+f(t\lambda_*))$ is at most $\varphi/d$ for
  all $t\in [0,1)$.  (See the example sketch in
  \cref{fig:close-to-imaginary-h}.)

  To prove this, we investigate the curve $\gamma(t)=1+f(t\lambda_*)$ for
  $t\in[0,1]$.  Note first that we have
  $0 \leq \arg(\gamma(t)) \leq \arg(f(t\lambda_*)) \leq \varphi$ for all
  $t \in [0,1]$.  Further, for all $t \in [0, 1)$, we have (here, we denote by
  $\gamma'(t)$ the right one-sided derivative of $\gamma$ at $t$)
  \begin{equation}
    \gamma'(t)=-\frac{d\lambda\lambda_*}{(1+t\lambda_*)^{d+1}}.\label{eq:29}
  \end{equation}
  Since $\varphi/d \ge \arg(1+t\lambda_*)\ge 0$ and
  $\arg(-\lambda\lambda_*)=2\varphi-\pi$, we have that as $t$ increases from $0$
  to $1$,
  \begin{equation}
    \arg \gamma'(t) =2\varphi-\pi-(d+1)\arg(1+t\lambda_*)\label{eq:35}
  \end{equation}
  decreases monotonically from $2\varphi-\pi \in [0, \pi)$ to $(1-1/d)\varphi-\pi \geq -\pi+\varphi/d$. 

  Next, we compute that
  $\diff{}{s} \arg\gamma(s)\vert_{s = t} = \Im \frac{\gamma'(t)}{\gamma(t)}$
  has the same sign as $\sin(\arg \gamma'(t) - \arg \gamma(t))$: note that the
  existence of this derivative follows since $\gamma(t)$ and $\gamma'(t)$ are
  non-zero, and since $\arg \gamma(t) \in [0, \pi)$ for $t \in [0, 1)$.

  We now claim that $\arg(\gamma(t)) \leq \varphi/d$ for all $t \in [0, 1]$.
  For the sake of contradiction let us assume that $\arg\gamma(t)$ can be bigger
  than $\varphi/d$. As $\arg\gamma(1)=0$, we then see that there must exist a
  $t_* \in [0, 1)$ such that $\arg \gamma(t_*) =\varphi/d$ and
  $\diff{}{s} \arg \gamma(s) \vert_{s = t_*} \leq 0$.  Using the fact (noted
  just below \cref{eq:35}) that $\pi > \arg \gamma'(t_*) \geq -\pi + \varphi/d$
  and the expression for the sign of $\diff{}{s} \arg(\gamma(s))\vert_{s = t_*}$
  noted above, these conditions can be written as
  \begin{equation}
    \arg \gamma(t_*) =\varphi/d \qquad \textrm{and}\qquad \arg\gamma'(t_*)
    \leq \varphi/d.\label{eq:36}
  \end{equation}
  Define $\alpha\defeq\arg(1+t_*\lambda_*)\geq 0$. \Cref{eq:36}, along with the
  expression for $\arg \gamma'(t)$ in \cref{eq:35}, and the fact
  $\arg \gamma(t_*) \leq \arg f(t_*\lambda_*) = \phi - d\alpha$ noted above,
  gives
  \begin{equation}
    \alpha \leq \frac{\varphi}{d}(1-1/d) \qquad \textrm{and}\qquad \alpha \geq
    \psi^*.\label{eq:32}
  \end{equation}
  The standard sine rule applied to the triangle with vertices $0,1,\gamma(t_*)$
  gives us that (see \cref{fig:triangle-sine-rule})
  \begin{equation}
    \frac{\abs{\gamma(t_*) - 1}}{\sin \arg \gamma(t_*)} =
    \frac{1}{\sin \inp{\arg \inp{\gamma(t_*) - 1} - \arg \gamma(t_*)} }. \label{eq:31}
  \end{equation}
  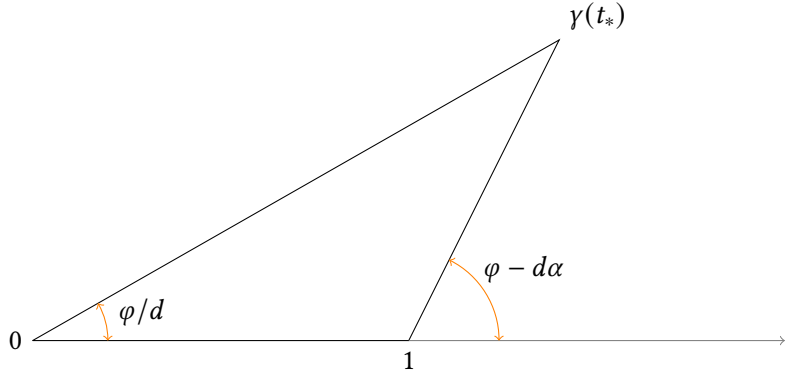
\begin{figure}[t]
    \centering
    \begin{tikzpicture}
      \draw
      (0, 0) coordinate (a) node[left] {$0$} -- (5,0) coordinate (b) node[below] {$1$} -- (7,4) coordinate (c) node[above right] {$\gamma(t_*)$} -- (a) pic["$\varphi/d$",draw=orange,<->,angle eccentricity=1.5,angle radius=1cm]
      {angle=b--a--c};
\draw [->, gray] (b)
      -- (10,0) coordinate (d) node[right] {};
\draw pic["$\varphi-d\alpha$",draw=orange,<->,angle eccentricity=1.5,angle
      radius=1.2cm] {angle=d--b--c};
    \end{tikzpicture}
    \caption{Angles appearing in \cref{eq:32}.}
    \label{fig:triangle-sine-rule}
  \end{figure}
  Using the facts that (i)
  $\abs{\gamma(t_*) - 1} = \abs{f(t_*\lambda_*)} = r/\abs{1 + t_*\lambda_*}^d$,
  (ii) $\arg\inp{\gamma(t_*) - 1} = \arg f(t_*\lambda_*) = \varphi - d\alpha$,
  and (iii) $\arg \gamma(t_*) = \varphi/d$, we get
  \begin{align}
    \nonumber \sin(\varphi/d)&=\frac{r\sin^d(\varphi-\alpha)}{\sin^d\varphi}\sin((1-1/d)\varphi-d\alpha)\\
                             &< \sin(\varphi/d)\frac{\sin^d(\varphi-\alpha)\sin((1-1/d)\varphi-d\alpha)}{\sin^d(\varphi-\psi^*)\sin((1-1/d)\varphi-d\psi^*)}. \label{eq:ineq}
  \end{align}
  But, in conjunction with \cref{eq:32}, this contradicts the fact that the
  function
  \begin{equation}
    x \mapsto \sin^d(\varphi-x)\sin((1-1/d)\varphi-dx)\label{eq:33}
  \end{equation}
  is a strictly decreasing function on $(\psi^*,(1-1/d)\varphi/d)$, since its
  derivative is
  \begin{equation}
    -{d\sin^{d-1}(\varphi-x)}\cdot\sin\inp{(2-1/d)\varphi - (d+1)x}<0.
    \label{eq:34}
  \end{equation}
  Here, we use the condition $x \in (\psi^*,(1-1/d)\varphi/d)$ and the
  definition of $\psi^*$ as
  $\max\left(\frac{1}{d+1}\left((2-1/d)\varphi-\pi\right),0\right)$ to deduce
  the last inequality.
\end{proof}

\section{Zero free regions in  the right  half plane}
\label{sec:right-half-plane}
In this section, we use the framework of \Cref{sec:crit-orig-compl} to establish
a zero free region for the independence polynomial in the right half plane.  The
results here improve upon those in the manuscript~\cite{bencs18:_note} when
$\lambda$ is close to the real axis and match those results when $\lambda$ is on
the imaginary axis: see \Cref{rem:right-improvement} for a more detailed
discussion.

We start with some notation.  For any integer $d \geq 2$, let
$\theta_d\in(\pi/(2(d+1)),\pi/2)$ be the unique solution of
\begin{equation}
  \label{eq:40}
  {\tan(2x/d)}=\frac{\tan((\pi/2-x)/d)}{1-\frac{\tan((\pi/2-x)/d)}{\tan(x)}}.
\end{equation}
To see that $\theta_d$ exists and is unique, we first note that the left hand
side of the above equation is monotone increasing while the right hand side is
monotone decreasing, so that it has at most one solution in the given interval.
To show existence, we note that as $x \downarrow \pi/(2(d+1))$, we have
\begin{equation}
  \label{eq:41}
  \lim_{x \downarrow \pi/(2(d+1))} \tan (2x/d) = \tan(\pi/(d(d+1))) < \lim_{x
    \downarrow \pi/(2(d+1))}
  \frac{\tan((\pi/2-x)/d)}{1-\frac{\tan((\pi/2-x)/d)}{\tan(x)}}=\infty,
\end{equation}
while as $x \uparrow \pi/2$ we have
\begin{equation}
  \label{eq:42}
  \lim_{x\uparrow\pi/2} \tan(2x/d) = \tan(\pi/d) >\lim_{x\uparrow\pi/2} \frac{
    \tan((\pi/2-x)/d) }{ 1 -\frac{\tan((\pi/2-x)/d)}{\tan(x)}} = 0.
\end{equation}
Together, these show that there is a unique solution $\theta_d$, such that for
all $x$ such that $\pi/(2(d+1)) < x < \theta_d$, we have
\begin{equation}
  \label{eq:10}
  \frac{\tan(2x/d)}{\sin x} < \frac{\tan ((\pi/2-x)/d)}{\sin(x) - \cos(x)\tan((\pi/2-x)/d)},
\end{equation}
while for $\theta_d < x < \pi/2$,
\begin{equation}
  \label{eq:11}
  \frac{\tan(2x/d)}{\sin x} > \frac{\tan ((\pi/2-x)/d)}{\sin(x) - \cos(x)\tan((\pi/2-x)/d)},
\end{equation}
For later comparison with results of \cite{bencs18:_note}, we also note that at
$x=\pi/6$ we have (assuming $d \geq 3$)

\begin{equation}
  \frac{\tan(2x/d)}{\sin x} = \frac{\tan(\pi/(3d))}{\sin (\pi/6)} <
  \frac{\tan(\pi/(3d))}{\sin(\pi/6)-\cos(\pi/6){\tan(\pi/(3d))}},\label{eq:37}
\end{equation}
which implies that $\pi/6<\theta_d$ for all $d \geq 2$ (this conclusion is
trivially true for $d = 2$).  We are now ready to state the main result
describing zeros in the right half plane.
\begin{theorem}\label{lem:righthalfplane_better}
  Let $\theta\in(0,\pi/2]$ and $0\le r\le r_{1,d}(\theta)$, where
  \begin{equation}
    \label{eq:43}
    r_{1,d}(\theta)=\left\{
      \begin{array}{cc}
        \frac{\tan(2\theta/d)}{\sin(\theta)}
        & \textrm{if $\theta \leq \theta_d$}\\
        \frac{\tan((\theta+\beta^*)/d)}{\sin(\theta)}
        & \textrm{if $\theta > \theta_d$}
      \end{array}
    \right. ,
  \end{equation}
  and where $\beta^*\in (0,\theta)$ is defined as the unique solution of
  \begin{equation}
    \label{eq:44}
    \frac{\tan((\theta+x)/d)}{\sin
      (\theta)}=\frac{\tan((\pi/2-\theta)/d)}{\sin(x)-\cos(x)\tan((\pi/2-\theta)/d)},
  \end{equation}
  when $\theta \in [\theta_d, \pi/2)$ and as $\beta^* \defeq 0$ when
  $\theta = \pi/2$.  If $\lambda=r\exp(\iota\theta)$, then $Z_G(\lambda)\neq 0$
  for any graph $G$ with degree at most $d+1$.
\end{theorem}
The proof of this theorem is based on the following technical lemma, which
employs the framework of \Cref{thm:init-curve}.

\begin{lemma}
  \label{lem:generic-right-half-plane}
  Let $\lambda=r\exp(\iota \theta)$ with $\theta\in(0,\pi/2]$ and $r > 0$.
  Suppose that there exist $r_2 \geq 0$ and $\beta, \psi \in[0,\pi/2)$ satisfying
 \begin{enumerate}
 \item $\theta-d\psi\ge -\beta$,\label{item:7}
 \item $r_2 \geq r$, \label{item:8}
 \item $r \sin(\theta) \le \tan \psi$,\label{item:9}
 \item $\theta+d\arg(1+r_2\exp(\iota \beta))\le \pi/2$,\label{item:10} and
 \item $\theta \geq \beta$.\label{item:11}
 \end{enumerate}
 Then the curve
 \begin{equation}
   \label{eq:45}
   h(t) \defeq
   \begin{cases}
     -t\cdot r_2\exp(-\iota\beta) &\textrm{if $t\in[-1,0]$}, and\\
     t\cdot \tan(\psi)\iota &\textrm{if $t\in[0,1]$}
   \end{cases}
 \end{equation}
 satisfies the conditions of \Cref{thm:init-curve}.
\end{lemma}

\begin{figure}[t]
  \centering
  \includegraphics[scale=0.6]{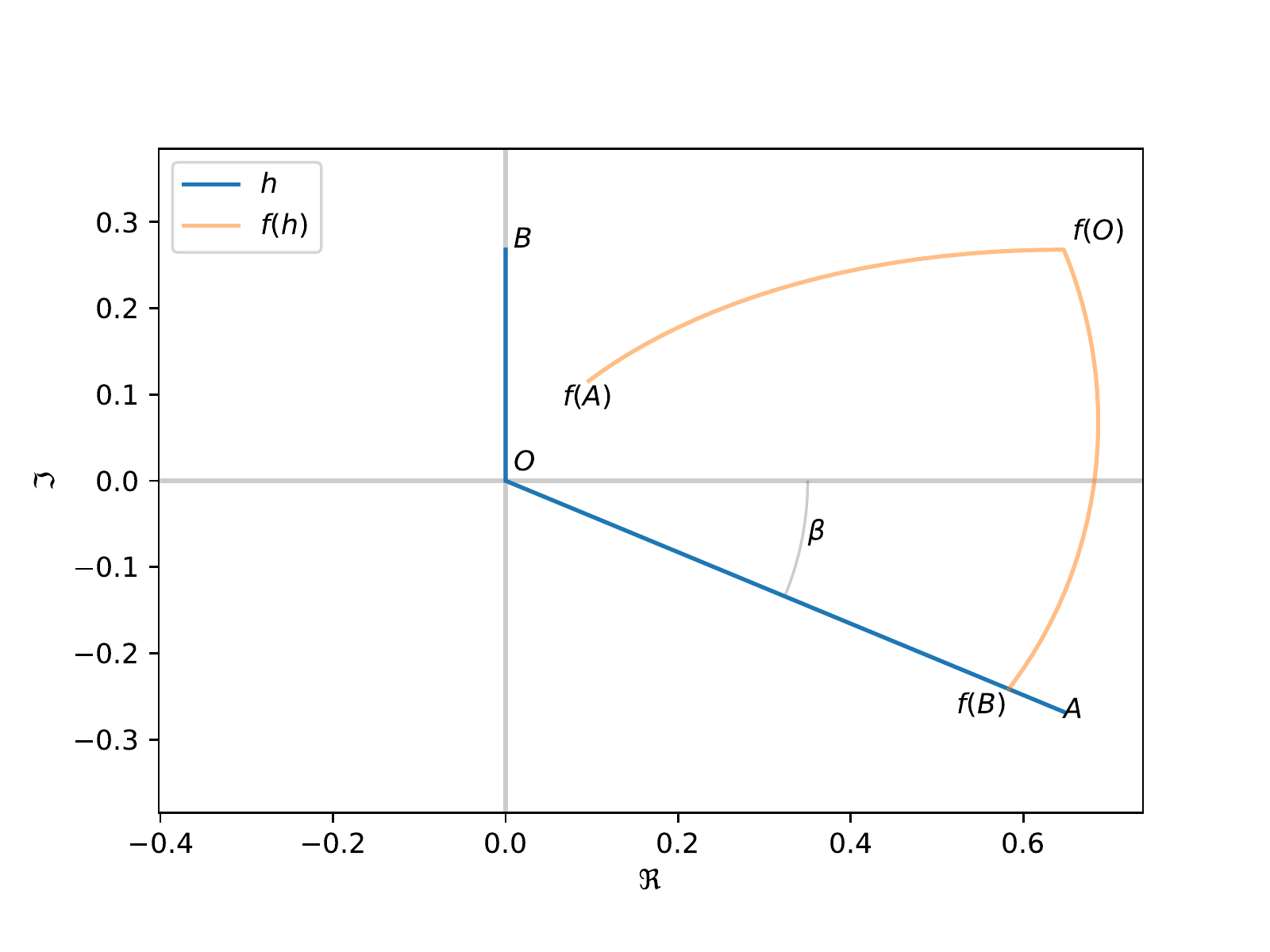}
  \caption{An example sketch of the $h$ curve from
    \cref{lem:generic-right-half-plane}.  In the notation of the lemma, the
    sketch corresponds to $d \defeq 3$,
    $\beta \defeq \theta \defeq \pi/(2(d+1))$, $\psi \defeq 2\theta/d$,
    $r_2 \defeq r \defeq \tan(2\theta/d)/\sin \theta$ and
    $\lambda \defeq r\exp(\iota \theta)$. As before, $f$ is the map
    $z \mapsto \lambda/(1+z)^d$.}
  \label{fig:right-half-plane-h}
\end{figure}

\begin{proof}
  Let $A$ denote the point $r_2\exp(-\iota\beta)$, $B$ the point
  $\iota \tan \psi$, and $O$ the origin. (See \cref{fig:right-half-plane-h} for
  an example sketch.)  Since $h(0) = 0$, $h$ satisfies the first condition of
  \Cref{thm:init-curve}.  Further, the curve $h(t)$ traverses the directed line
  segment $AO$ as $t$ varies from $-1$ to $0$ and the line segment $OB$ as $t$
  varies from $0$ to $1$, and this establishes the second condition of
  \Cref{thm:init-curve} (since $\beta \in [0, \pi/2)$ and $\psi \geq 0$).

  A convex combination of any two points $h(t_1)$ and $h(t_2)$, where
  $t_1 < t_2$, lies either on the curve $h$ (when $0 \not\in (t_1, t_2)$), or
  on the boundary of the triangle with vertices $h(t_1) \in AO, h(0) = 0$ and
  $h(t_2) \in OB$ (when $0 \in (t_1, t_2)$).  It is therefore $-1$ covered by
  $h(t)$ for some $t \in [t_1, t_2$].  This establishes the third condition of
  \Cref{thm:init-curve}.

  Note that the segment $OB$ of the curve $h$ $-1$-covers every point in the set
  \begin{equation}
    \label{eq:46}
    L_1\defeq \inb{z: \Re z \geq 0 \text{ and } 0 \leq \Im z \leq \Im B = \tan
      \psi \geq 0},
  \end{equation}
  while the segment $AO$ $-1$-covers every point in the set
  \begin{equation}
    \label{eq:47}
    L_2\defeq \inb{z: -\pi/2 < -\beta \leq \arg z \leq 0 \text{ and } \Im z \geq
      \Im A = -r_2\sin \beta \leq 0}.
  \end{equation}

  It remains to verify the fourth condition of \Cref{thm:init-curve}, which is
  that for every $t \in [0, 1]$, $f(h(t))$ is $-1$-covered by some point on the
  curve $h$.  We do so by proving that for all $t \in [-1, 1]$,
  $f(h(t)) \in L_1 \cup L_2$.

  Consider first a point $f(h(t))$ for $t \in [0, 1]$.  Note that as $t$
  increases from $0$ to $1$, $\arg(1 + h(t))$ increases from $0$ to $\psi$.
  From \cref{item:7} in the statement of the lemma, we thus get that
  $\arg f(h(t)) \in [-\beta, \theta]$, while \cref{item:8} gives
  $\abs{f(h(t))} \leq r \leq r_2$.  Together with \cref{item:9}, these imply
  that for $t \in [0, 1]$,
  \begin{enumerate}
  \item $0 \leq \Im f(h(t)) \leq r\sin \theta \leq \tan \psi$ (when
    $\arg f(h(t)) \geq 0$), and
  \item $0 \geq \Im f(h(t)) \geq -r\sin\beta \geq -r_2\sin\beta$ (when
    $\arg f(h(t)) \leq 0$).
  \end{enumerate}
  Thus, for all $t \in [0, 1]$, $f(h(t)) \in L_1 \cup L_2$, and thus is
  $-1$-covered by the curve $h$.

  Now, consider a point $f(h(t))$ for $t\in [-1, 0]$.  Define
  $g(t) \defeq f(h(t))$.  From \cref{item:10}, we get that
  $\arg g(t) \in [\theta, \pi/2]$ for all $t \in [-1, 0]$. We also have
  $g(0) = \lambda \in L_1$ (where the last inclusion follows from
  \cref{item:9}).  Thus, in order to establish that $g(t) \in L_1$ for all
  $t \in [-1, 0]$, it suffices to prove that $k(t) \defeq \Im g(t)$ has a
  non-negative right derivative at every $t \in [-1, 0)$.  The latter in turn
  would follow if we establish that $\arg g'(t) \in [0, \pi)$ for all
  $t \in [-1, 0)$, where $g'(t)$ denotes the right derivative of $g$ at $t$.

  We now compute, for $t \in [-1, 0)$,
  \begin{equation}
    \label{eq:48}
    g'(t) = d \cdot r_2\cdot\frac{\lambda}{(1 - t\cdot r_2\exp(-\iota \beta))^d}
    \frac{\exp(-\iota\beta)}{(1 - t\cdot r_2\exp(-\iota \beta))}
  \end{equation}
  so that (after multiplying denominators with conjugates and ignoring positive
  real factors)
  \begin{equation}
    \label{eq:49}
    \arg g'(t) = \arg\inp{\underbrace{\insq{\exp(\iota \theta) \cdot (1 - t\cdot
          r_2\exp(\iota \beta))^d}}_{=: \mu} \cdot \underbrace{ \insq{\exp(-\iota
          \beta) - t \cdot r_2} }_{=:\nu} }.
 \end{equation}
 Since $t \in [-1, 0)$, \cref{item:10} in the statement of the lemma then
 implies that $\arg \mu \in [\theta, \pi/2]$.  Further,
 $\arg \nu \in [-\beta, 0]$.  Together with $\theta \geq \beta$ (\cref{item:11}
 in the statement of the lemma), this implies that
 $\arg g'(t) = \arg\inp{\mu \cdot \nu} \in [0, \pi/2] \subseteq [0, \pi)$.
 Given the above discussion about the relationship between the functions $k$ and
 $g$, this completes the proof.
\end{proof}

With the above lemma, we can now complete the proof of
\Cref{lem:righthalfplane_better}.

\begin{proof}[Proof of \Cref{lem:righthalfplane_better}]
  We will prove that if $r\leq r_{1,d}(\theta)$, then we can find
  $\beta, \psi\in[0,\pi/2)$ and $r_2>0$, such that the conditions of
  \Cref{lem:generic-right-half-plane} hold. By applying
  \cref{thm:init-curve} for the curve $h$ obtained from the lemma, we get
  the desired statement.
  \begin{enumerate}
  \item Consider first the case $\theta \leq \theta_d$. Then, let
    \begin{equation}
      \beta=\theta,~\qquad \psi=\frac{2\theta}{d}, \qquad
      r_2=\argmax\limits_{t\ge 0} \left(\arg(1+t\exp(\iota \theta))\le
        \frac{\pi/2-\theta}{d}\right)\le\infty.\label{eq:21}
    \end{equation}
    \Cref{item:7,item:9,item:10,item:11} in \Cref{lem:generic-right-half-plane}
    are satisfied by construction (as discussed below, we might have to redefine
    $r_2$ to make sure it is finite). We now show that \cref{item:8} holds:
    \begin{itemize}
    \item If $\theta \leq \pi/(2(d+1))$, then $r_2=\infty\ge r$.  In this case,
      we redefine $r_2 = r$, and all of the conditions continue to hold.
    \item Otherwise $\theta> \pi/(2(d+1))$.  In this case, we have
      $r_2=\frac{\tan((\pi/2-\theta)/d)}{\sin(\theta)-\cos(\theta)\tan((\pi/2-\theta)/d)}
      \geq \frac{\tan(2\theta/d)}{\sin(\theta)}=r_{1,d}(\theta) \geq r$, where
      the first inequality follows from \cref{eq:10} since
      $\pi/(2(d+1)) < \theta \leq \theta_d$.
    \end{itemize}
  \item Consider now the case $\theta >\theta_d$. Then let $\beta$ be
    $\beta^*\in [0,\theta]$ as described in the statement of the theorem.  By
    definition, $\beta^* = 0$ when $\theta = \pi/2$, so we first show that even
    when $\theta \in (\theta_d, \pi/2)$, this $\beta^*$ exists and is unique. To
    see this, note that
    \begin{equation}
      \gamma_1(x) \defeq \frac{\tan((\theta+x)/d)}{\sin (\theta)}\label{eq:22}
    \end{equation}
    is continuous, monotone increasing and positive on $[0,\theta]$ (when
    $0 < \theta < \pi/2$).  On the other hand,
    \begin{equation}
      \gamma_2(x) \defeq
      \frac{\tan((\pi/2-\theta)/d)}{\sin(x)-\cos(x)\tan((\pi/2-\theta)/d)}\label{eq:23}
    \end{equation}
    is continuous in
    $\left[0, (\pi/2-\theta)/d)\right) \cup \left((\pi/2-\theta)/d,
      \theta\right]$, negative in $\left[0, (\pi/2-\theta)/d)\right)$, and
    monotone decreasing in $\left((\pi/2-\theta)/d,\theta\right]$.  Further, in
    the interval $\left((\pi/2-\theta)/d,\theta\right]$ we also have
    \begin{equation}
      \gamma_1((\pi/2-\theta)/d) <\infty,
      \qquad
      \textrm{and}
      \qquad
      \lim_{x\downarrow (\pi/2-\theta)/d}\gamma_2(x) = \infty,\label{eq:24}
    \end{equation}
    at the left endpoint, while at the right endpoint, $\theta>\theta_d$ implies
    $\gamma_1(\theta) > \gamma_2(\theta)$ (due to \cref{eq:11}).  The above
    observations imply that when $\theta \in (\theta_d, \pi/2)$,
    $\gamma_1(x) = \gamma_2(x)$ has exactly one solution
    $\beta^* \in [0, \theta]$, which lies in $((\pi/2-\theta)/d, \theta]$.

    Then, we define
    \begin{equation}
      \begin{gathered}
        \beta = \beta^*, \qquad \psi=(\theta+\beta)/d,\\
        r_2=
        \begin{cases}
          \argmax\limits_{t\ge 0} \left(\arg(1+t\exp(\iota \beta))\le
            \frac{\pi/2-\theta}{d}\right)<\infty & \text{ when } \theta \in
          (\theta_d, \pi/2),\\
          \tan (\pi/(2d)) & \text{ when } \theta = \pi/2.
        \end{cases}
      \end{gathered}\label{eq:14}
  \end{equation}
  Note that $r_2$ is finite when $\theta \in (\theta_d, \pi/2)$ since
    $\beta = \beta^* > (\pi/2-\theta)/d$.

    Again, \cref{item:7,item:9,item:10,item:11} in
    \Cref{lem:generic-right-half-plane} are satisfied by construction. We now
    show that \cref{item:8} holds.  To see this, we first note that when
    $\theta_d < \theta < \pi/2$, \cref{item:8} holds since in that case,
    \cref{eq:14} and the definition of $\beta^*$ give
    $r_2=\frac{\tan((\pi/2-\theta)/d)}{\sin(\beta^*)-\cos(\beta^*)\tan((\pi/2-\theta)/d)}=
    \frac{\tan((\theta+\beta^*)/d)}{\sin(\theta)}=r_{1,d}(\theta) \geq r$. In
    the remaining case $\theta = \pi/2$, \cref{item:8} holds since in that case,
    \cref{eq:14} gives again $r_2 = \tan(\pi/(2d)) = r_{1,d}(\pi/2) \geq r$.
    \qedhere
 \end{enumerate}
\end{proof}

\begin{remark}\label{rem:right-improvement}
  We remark that the zero-free region established in
  \Cref{lem:righthalfplane_better} contains the zero-free region described in
  the manuscript \cite{bencs18:_note} when $\arg \lambda = \theta \leq \theta_d$
  (recall also from the paragraph just before the statement of
  \Cref{lem:righthalfplane_better} that $\theta_d$ is always greater than
  $\pi/6$).  For such $\theta$, the above theorem gives zero-freeness for all
  $\lambda$ with $\abs{\lambda} < \tan(2\theta/d)/\sin(\theta)$ and
  $\arg(\lambda) = \theta$.  On the other hand, the zero-free region in Theorem
  1.4 of~\cite{bencs18:_note} requires at least that
  $\abs{\lambda} \leq \tan(\pi/(2d))$.  But when $d \geq 2$ and
  $\theta \in (0, \pi/2)$, elementary arguments involving the convexity of the
  function $\theta \mapsto \tan(2\theta/d) - \sin(\theta)\tan(\pi/(2d))$ in the
  interval $(0, \pi/ 2)$ imply that
  $\tan(\pi/(2d))< \tan(2\theta/d)/\sin(\theta)$, showing that
  \cref{lem:righthalfplane_better} gives a larger zero-free region.  For the
  case $\theta = \pi/2$, we compute directly that
  $r_{1, d}(\pi/2) = \tan(\pi/(2d))$.

  For the case $\theta_d < \arg \lambda < \pi/2$, the zero-free region in
  \cite{bencs18:_note} has only an implicit description, and numerical
  calculations show that even in this case, the zero free region described in
  \cref{lem:righthalfplane_better} is better than the one in
  \cite{bencs18:_note}, except possibly in the close vicinity of
  $\theta = \pi/2$.

\end{remark}

\bibliography{zeros}

\begin{thebibliography}{10}

\bibitem{arora_barak_2009}
{\sc Arora, S., and Barak, B.}
\newblock {\em Computational Complexity: A Modern Approach}.
\newblock Cambridge University Press, 2009.

\bibitem{bandyopadhyay_counting_2008}
{\sc Bandyopadhyay, A., and Gamarnik, D.}
\newblock Counting without sampling: {Asymptotics} of the log-partition
  function for certain statistical physics models.
\newblock {\em Random Structures \& Algorithms 33}, 4 (2008), 452--479.
\newblock Extended abstract in SODA 2006.

\bibitem{barvinok_computing_2015}
{\sc Barvinok, A.}
\newblock Computing the {Partition} {Function} for {Cliques} in a {Graph}.
\newblock {\em Theory Comput. 11\/} (Dec. 2015), 339--355.

\bibitem{barvinok15:_comput_perman_some_compl_matric}
{\sc Barvinok, A.}
\newblock Computing the {Permanent} of ({Some}) {Complex} {Matrices}.
\newblock {\em Found. Comput. Math. 16}, 2 (Jan. 2015), 329--342.

\bibitem{barvinok2017combinatorics}
{\sc Barvinok, A.}
\newblock {\em Combinatorics and Complexity of Partition Functions}.
\newblock Algorithms and Combinatorics. Springer, 2017.

\bibitem{BarvinokSoberon16a}
{\sc Barvinok, A., and Sober{\'{o}}n, P.}
\newblock Computing the partition function for graph homomorphisms.
\newblock {\em Combinatorica 37}, 4 (2017), 633--650.

\bibitem{bencs_trees_2018}
{\sc Bencs, F.}
\newblock On trees with real-rooted independence polynomial.
\newblock {\em Disc. Math. 341}, 12 (Dec. 2018), 3321--3330.

\bibitem{bencs2021limit}
{\sc Bencs, F., Buys, P., and Peters, H.}
\newblock The limit of the zero locus of the independence polynomial for
  bounded degree graphs, Nov. 2021.
\newblock \arxiv{2111.06451}.

\bibitem{bencs18:_note}
{\sc {Bencs}, F., and {Csikv{\'a}ri}, P.}
\newblock Note on the zero-free region of the hard-core model, Jul 2018.
\newblock \arxiv{1807.08963}.

\bibitem{bencs2018zero}
{\sc Bencs, F., Davies, E., Patel, V., and Regts, G.}
\newblock On zero-free regions for the anti-ferromagnetic {Potts} model on
  bounded-degree graphs.
\newblock {\em Ann. Inst. Henri Poincaré D 8}, 3 (Sept. 2021), 459--489.

\bibitem{bezakova_inapproximability_2018}
{\sc Bezáková, I., Galanis, A., Goldberg, L.~A., and Štefankovič, D.}
\newblock Inapproximability of the {Independent} {Set} {Polynomial} in the
  {Complex} {Plane}.
\newblock {\em SIAM J. Comput. 49}, 5 (2020), STOC18--395--STOC18--448.
\newblock STOC 2018 Special Section.

\bibitem{buys_location_2019}
{\sc Buys, P.}
\newblock Cayley {Trees} do {Not} {Determine} the {Maximal} {Zero}-{Free}
  {Locus} of the {Independence} {Polynomial}.
\newblock {\em Michigan Math. J. 70}, 3 (Aug. 2021), 635--648.

\bibitem{de2021zeros}
{\sc de~Boer, D., Buys, P., Guerini, L., Peters, H., and Regts, G.}
\newblock Zeros, chaotic ratios and the computational complexity of
  approximating the independence polynomial, Apr. 2021.
\newblock \arxiv{2104.11615}.

\bibitem{DS85}
{\sc Dobrushin, R.~L., and Shlosman, S.~B.}
\newblock Completely {Analytical} {Gibbs} {Fields}.
\newblock In {\em Statistical {Physics} and {Dynamical} {Systems}: {Rigorous}
  {Results}}, J.~Fritz, A.~Jaffe, and D.~Szász, Eds., Progress in {Physics}.
  Birkhäuser, Boston, MA, 1985, pp.~371--403.

\bibitem{DS87}
{\sc Dobrushin, R.~L., and Shlosman, S.~B.}
\newblock {Completely Analytical Interactions: Constructive Description}.
\newblock {\em J. Stat. Phys. 46\/} (1987), 983--1014.

\bibitem{eldar2018approximating}
{\sc Eldar, L., and Mehraban, S.}
\newblock Approximating the permanent of a random matrix with vanishing mean.
\newblock In {\em Proceedings of the 59th Annual IEEE Symposium on Foundations
  of Computer Science (FOCS)\/} (2018), IEEE, pp.~23--34.

\bibitem{Vigoda-hard-core-11}
{\sc Galanis, A., Ge, Q., {\v{S}}tefankovi{\v{c}}, D., Vigoda, E., and Yang,
  L.}
\newblock Improved inapproximability results for counting independent sets in
  the hard-core model.
\newblock {\em Random Struct. Algorithms 45}, 1 (2014), 78--110.
\newblock Extended abstract in APPROX-RANDOM 2011.

\bibitem{gsv}
{\sc Galanis, A., {\v{S}}tefankovi\v{c}, D., and Vigoda, E.}
\newblock Inapproximability for antiferromagnetic spin systems in the tree
  nonuniqueness region.
\newblock {\em J. ACM 62}, 6 (2015), 50:1--50:60.
\newblock Extended abstract in STOC 2014.

\bibitem{harrow_classical_2020}
{\sc Harrow, A.~W., Mehraban, S., and Soleimanifar, M.}
\newblock Classical algorithms, correlation decay, and complex zeros of
  partition functions of {Quantum} many-body systems.
\newblock In {\em Proceedings of the 52nd {Annual} {ACM} {Symposium} on
  {Theory} of {Computing} (STOC)}. ACM, June 2020, pp.~378--386.

\bibitem{10.5555/3174304.3175407}
{\sc Harvey, N. J.~A., Srivastava, P., and Vondr\'{a}k, J.}
\newblock Computing the independence polynomial: From the tree threshold down
  to the roots.
\newblock In {\em Proceedings of the 29th Annual ACM-SIAM Symposium on Discrete
  Algorithms (SODA)\/} (2018), SIAM, pp.~1557--1576.

\bibitem{li_complex_2021}
{\sc Li, L., and Xie, G.}
\newblock Complex contraction on trees without proof of correlation decay, Dec.
  2021.
\newblock \arxiv{2112.15347}.

\bibitem{JingchengThesis}
{\sc Liu, J.}
\newblock {\em Approximate Counting, Phase Transitions and Geometry of
  Polynomials}.
\newblock PhD Thesis, UC Berkeley, 2019.
\newblock Available at
  \url{https://www2.eecs.berkeley.edu/Pubs/TechRpts/2019/EECS-2019-110.html}.

\bibitem{liu19:_deter_algor_count_color_delta_color}
{\sc Liu, J., Sinclair, A., and Srivastava, P.}
\newblock A deterministic algorithm for counting colorings with {$2\Delta$}
  colors.
\newblock In {\em Proceedings of the 60th Annual IEEE Symposium on Foundations
  of Computer Science (FOCS)\/} (2019), IEEE, pp.~1380--1404.

\bibitem{liu2018fisher}
{\sc Liu, J., Sinclair, A., and Srivastava, P.}
\newblock Fisher zeros and correlation decay in the {Ising} model.
\newblock {\em J. Math. Phys. 60\/} (2019), 103304.
\newblock Extended abstract in proceedings of ITCS 2019.

\bibitem{patel2017deterministic}
{\sc Patel, V., and Regts, G.}
\newblock Deterministic polynomial-time approximation algorithms for partition
  functions and graph polynomials.
\newblock {\em SIAM J. Comput. 46}, 6 (2017), 1893--1919.

\bibitem{peters_conjecture_2019}
{\sc Peters, H., and Regts, G.}
\newblock On a {Conjecture} of {Sokal} {Concerning} {Roots} of the
  {Independence} {Polynomial}.
\newblock {\em Michigan Math. J. 68}, 1 (2019), 33--55.

\bibitem{scott_repulsive_2005}
{\sc Scott, A.~D., and Sokal, A.~D.}
\newblock The {Repulsive} {Lattice} {Gas}, the {Independent-Set} {Polynomial},
  and the {Lovász} {Local} {Lemma}.
\newblock {\em J. Stat. Phys. 118}, 5-6 (Mar. 2005), 1151--1261.

\bibitem{shaosun19}
{\sc Shao, S., and Sun, Y.}
\newblock Contraction: {A} {Unified} {Perspective} of {Correlation} {Decay} and
  {Zero}-{Freeness} of 2-{Spin} {Systems}.
\newblock {\em J. Stat. Phys. 185}, 2 (Oct. 2021), 12.
\newblock Extended abstract in proceedings of ICALP 2020.

\bibitem{shearer_problem_1985}
{\sc Shearer, J.~B.}
\newblock On a problem of {Spencer}.
\newblock {\em Combinatorica 5}, 3 (Sept. 1985), 241--245.

\bibitem{sipser_complexity_1983}
{\sc Sipser, M.}
\newblock A complexity theoretic approach to randomness.
\newblock In {\em Proceedings of the 15th Annual {ACM} Symposium on {Theory} of
  Computing (STOC)\/} (1983), ACM, pp.~330--335.

\bibitem{Sly2010CompTransition}
{\sc Sly, A.}
\newblock Computational transition at the uniqueness threshold.
\newblock In {\em Proceedings of the 51st Annual IEEE Symposium on Foundations
  of Computer Science (FOCS)\/} (2010), IEEE, pp.~287--296.

\bibitem{sly12}
{\sc Sly, A., and Sun, N.}
\newblock Counting in two-spin models on $d$-regular graphs.
\newblock {\em Ann. Probab. 42}, 6 (2014), 2383--2416.

\bibitem{stockmeyer_complexity_1983}
{\sc Stockmeyer, L.}
\newblock The complexity of approximate counting.
\newblock In {\em Proceedings of the 15th Annual {ACM} Symposium on {Theory} of
  Computing (STOC)\/} (1983), ACM, pp.~118--126.

\bibitem{toda_pp_1991}
{\sc Toda, S.}
\newblock {PP} is as {Hard} as the {Polynomial}-{Time} {Hierarchy}.
\newblock {\em SIAM J. Comput. 20}, 5 (Oct. 1991), 865--877.

\bibitem{Weitz}
{\sc Weitz, D.}
\newblock Counting independent sets up to the tree threshold.
\newblock In {\em Proceedings of the 38th Annual ACM Symposium on Theory of
  Computing (STOC)\/} (2006), ACM, pp.~140--149.

\bibitem{leeyan52}
{\sc Yang, C.~N., and Lee, T.~D.}
\newblock Statistical theory of equations of state and phase transitions. {I.
  Theory} of condensation.
\newblock {\em Phys. Rev. 87}, 3 (1952), 404--409.

\end{thebibliography}
\bibliographystyle{acm}

\end{document}